\documentclass[a4paper,12pt]{article}
\usepackage{amssymb,latexsym,amsmath,amscd,amsthm}
\usepackage{comment}
\usepackage{epsfig}
\usepackage{color}
\usepackage[all]{xy}
\usepackage{pgf}
\usepackage{url}
\usepackage{comment}

\newtheorem{thm}{Theorem}[section]
\newtheorem{lema}[thm]{Lemma}
\newtheorem{corollary}[thm]{Corollary}
\newtheorem{prop}[thm]{Proposition}

\newtheorem{ex}[thm]{Example}

\theoremstyle{definition}

\newtheorem{rmk}[thm]{Remark}
\newenvironment{prooftripod}{\noindent {\textit{Proof of Proposition 4.2.}}}{$\square$ \vspace{3mm}}

\newcommand{\N}{\mathcal{N}}
\newcommand{\LL}{\mathcal{L}}

\newcommand{\MM}{\mathcal{M}}
\newcommand{\F}{\mathcal{F}}

\newcommand{\ZZ}{\mathbb{Z}}
\newcommand{\CC}{\mathbb{C}}
\newcommand{\Int}{\mathrm{Int}}

\newcommand{\op}[2]{#1 \cdot #2}
\newcommand{\ddd}[1]{\mathrm{diag}(#1)\,}

\newcommand{\X}{\texttt{X}}

\renewcommand{\a}{\texttt{A}}
\renewcommand{\c}{\texttt{C}}
\newcommand{\g}{\texttt{G}}
\renewcommand{\t}{\texttt{T}}
\newcommand{\ba}{\bar{\texttt{A}}}
\newcommand{\bc}{\bar{\texttt{C}}}
\newcommand{\bg}{\bar{\texttt{G}}}
\newcommand{\bt}{\bar{\texttt{T}}}
\newcommand{\bp}{\bar{p}}

\newcommand{\rk}{\mathbf{rk}\,}
\newcommand{\rank}{\mathrm{rank}\,}

\newcommand{\Hom}{\mathrm{Hom}\,}

\renewcommand{\Im}{\mathrm{Im}\,}

\usepackage{color}

\definecolor{pink}{rgb}{1,0,1}

\definecolor{garnet}{RGB}{210,15,30}

\title{Phylogenetic Invariants: From the General Markov to Equivariant Models}

\author{Marta Casanellas, Jes\'us Fern\'andez-S\'anchez\footnote{Universitat Polit{\`e}cnica de Catalunya and Centre de Recerca Matem\`atica. ETSEIB, Av.~Diagonal 647, 08028 Barcelona. E-mail addresses: \texttt{marta.casanellas@upc.edu} and \texttt{jesus.fernandez.sanchez@upc.edu}}}
\date{}
\begin{document}

\maketitle

\abstract{
In the last decade, some algebraic tools have been successfully applied to phylogenetic reconstruction. 
These tools are mainly based on the knowledge of equations describing algebraic varieties associated to  Markov processes of molecular substitution on phylogenetic trees, the so called \emph{phylogenetic invariants}. Although the theory involved allows for the explicit determination of these equations for all \emph{equivariant models} (which include some of the most popular nucleotide substitution models), practical uses of these algebraic tools have been restricted to the case of the general Markov model. Arguably,  one of the reasons for this restriction is that knowledge of linear representation theory is required before making these equations explicit.

With the aim of enlarging the practical uses of algebraic phylogenetics, in this paper we prove that phylogenetic invariants for equivariant models can be derived from phylogenetic invariants for the general Markov model, without the need of representation theory. Our main result states that the algebraic variety corresponding to an equivariant model on a phylogenetic tree $T$ is an irreducible component of the variety corresponding to the general Markov model on $T$ intersected with the linear space defined by the model. We also prove that, for any equivariant model, those phylogenetic invariants that are relevant for practical uses (e.g. tree reconstruction) can be simply deduced from a single rank constraint on the matrices obtained by flattening the joint distribution at the leaves of the tree. This condition can be easily tested from singular values of the matrices and it allows us to extend our results from trees to phylogenetic networks.
}

\section{Introduction}
Phylogenetics aims at reconstructing the evolutionary history of a set of species (or other biological entities) from molecular data. This evolutionary history is usually represented by a \emph{phylogenetic tree} whose leaves stand for currently living species and whose interior nodes correspond to their ancestral species. Molecular data is commonly given as a sequence of characters representing nucleotides or amino acids and phylogenetic reconstruction is often done by modelling the substitution of these characters as a hidden Markov process on a phylogenetic tree.

In the late eighties, Cavender, Felsenstein, and Lake realized  that polynomial equations satisfied by the entries of the joint distribution of characters at the leaves of the tree could be used in phylogenetic reconstruction, see \cite{Cavender87,Lake1987}. By then, only  few polynomial equations were known and exclusively for very simple models such as the Kimura 2-parameter model \cite{Kimura1980}. The use of these equations known as \emph{phylogenetic invariants} was set aside until the beginning of the new century. In the last twenty years there has been a lot of effort to obtain phylogenetic invariants for different evolutionary models: Allman and Rhodes worked on obtaining equations for the general Markov model \cite{allman2004b, Allman2008}, Sturmfels and Sullivant provided phylogenetic invariants for group-based models \cite{Sturmfels2005}, Draisma and Kuttler generalized the work done by Allman and Rhodes to equivariant models (which include group-based models) in \cite{Draisma}, and many others contributed to specific models or trees (see for example \cite{CS}, \cite{ChifmanPetrovic}, \cite{michalek2013}, \cite{LS}).


The main tool that has allowed practical application of these phylogenetic invariants has been its translation into rank conditions of certain matrices arising from \emph{flattening} the joint distribution according to certain bipartitions of the set of leaves.  By the Eckart-Young theorem  (see \cite{Eckart1936}), the distance of a matrix to the set of matrices of a given rank can be easily computed from the least singular values.
This approach has been used in practice, at least for phylogenetic reconstruction based on quartets, in the case of Jukes-Cantor and GTR under the coalescent model in the SVDquartets method of species tree inference (see \cite{chifmankubatko2014}), and for the general Markov model in \cite{Erik2} and  \cite{casfergar23}. 
Although the general Markov model (which arises when no constraints are imposed on the transition probabilities or the root distribution of a hidden Markov process on a phylogenetic tree) is reasonable for nucleotide data, it {seems} too general when dealing with amino acid data {(each transition matrix involves 380 free parameters in this case).} The main obstacle to implementing the invariants found by Draisma and Kuttler \cite{Draisma} for any equivariant model is that some knowledge of linear representation theory is needed to {translate the rank conditions into explicit equations that can be evaluated on the empirical data}.
With the goal of making algebraic phylogenetics practical for equivariant models, in this work we present a novel approach to obtain phylogenetic invariants for any equivariant model from those of the general Markov model.


We describe our approach in what follows. If $G$ is any permutation subgroup of a set of $\kappa$ states  ($\kappa$=4 for nucleotides and 20 for amino acids), a $G$-equivariant model on a phylogenetic tree $T$ is defined by imposing that the transition matrices are $G$-equivariant (equivalently, they remain invariant by the permutation action of $G$ on rows and columns) and that the root distribution is invariant by permutations in $G$. These models include the well known Kimura with two (K80) and three parameters (K81) \cite{Kimura1981}, the Jukes-Cantor (JC69) \cite{JC69}, the strand symmetric model \cite{CS}, and the general Markov model (when $G$ is the trivial group). The set of distributions at the leaves of a phylogenetic tree $T$ that arise as a hidden Markov process on $T$ under the restrictions of a $G$-equivariant model lies in an algebraic variety $V_T^G$ defined as the Zariski closure of this set of distributions. Phylogenetic invariants mentioned above are polynomials in the ideal of these algebraic varieties.

For a given set of leaves $L$, the algebraic varieties $V_T^G$ for different phylogenetic trees $T$ with leaf set $L$ naturally lie in the same linear space $\mathcal{L}^G$, which contains those distributions that are invariant by the permutations in $G$. In the work \cite{Draisma}, the authors gave a procedure for obtaining the equations that define $V_T^G$ inside $\LL^G$ from the equations of tripod trees and rank conditions on block-diagonal flattening matrices. This involves a change of basis that requires some knowledge of representation theory and decomposing vector spaces into isotypic components. As proved in \cite{CF11} and in \cite{CFM}, this turns out to be a tedious task for each particular tree $T$ and permutation group $G$ and makes these rank conditions  impractical for phylogenetic inference. In contrast, it is very easy to obtain the equations of $\LL^G$ and the rank conditions of flattenings for the variety  $V_T$ corresponding to the general Markov model.
So one basic question in algebraic phylogenetics appears: as the parameters for a $G$-equivariant model correspond to linear constraints on the parameters of the general Markov model, it is natural to ask whether $V_T^G$ is a linear section of $V_T$. Actually, the proper question is the following:

\vspace*{1mm}
\textbf{Question 1:} Is $V_T^G$ equal to $V_T\cap \mathcal{L}^G$?
\vspace*{1mm}

If this question had a positive answer, then finding equations for $V_T^G$ would be a simple task. But the answer to Question 1 is negative in general (see section 3.2).
Nevertheless, in our main result (Theorem \ref{prop_main}) we prove that $V_T^G$ is an irreducible component of $V_T\cap \mathcal{L}^G$. This implies that equations of $V_T$ and $\mathcal{L}^G$ are enough for describing $V_T^G$.
The result is proven by adapting the proof of a result of Chang on the identifiability of parameters for the general Markov model on phylogenetic trees, see \cite{chang1996}.

{In \cite{CF11} we proved that Draisma-Kuttler equations from block-diagonal rank conditions (in a basis adapted to isotypic components) are enough to define, on an open set, the variety $V_T^G$ inside the union of all varieties $\cup_T V_T^G$ (where the union runs over all trivalent trees in the same set of leaves); in other words, these equations suffice to detect the tree $T$ once we know that the distribution has arisen under a $G$-equivariant process on some tree. In Theorem \ref{thm_edges} we prove that these equations  can be simply reduced to imposing rank $\leq \kappa$ on the usual flattening matrix.
One of the main consequences of this result is that there is no need to decompose the flattening into isotypic components to detect the variety $V_T^G$:  $(\kappa+1)\times(\kappa +1)$ minors on the usual flattening will be enough to define $V_T^G$ inside the union  $\cup_T V_T^G$ (on an open set).}
We prove Theorem \ref{thm_edges} for tensors in general (not only for those arising from processes on phylogenetic trees) and as a byproduct we obtain phylogenetic invariants for $G$-equivariant models on {certain}  phylogenetic networks. We expect that this might have consequences on the identifiability of phylogenetic networks.

Note that all results in this paper work for any number of states, which implies that they {could} be used to obtain phylogenetic invariants for amino acid $G$-equivariant models. 
{Although this approach for amino acid substitution models  has not yet been explored, there are some models with symmetries which hint that this may be an amenable approach both for amino acid or codon models \cite{Shore2020,Kubatko2016}.}

The organization of the paper is as follows. In section 2 we introduce notation and the preliminary material needed: Markov processes on trees, $G$-equivariant models, flattenings and phylogenetic algebraic varieties. In section 3 we motivate our work by exploring some basic examples (namely, tripods and quartets with JC69, K80 or K81 models); we give the first negative answer to Question 1 but we also shed some light on the study of $V_T\cap \LL^G$. In section 4 we prove our main result
Theorem \ref{prop_main} by using techniques from linear algebra following and generalizing the approach taken by Chang \cite{chang1996}. In section 5 we introduce techniques from representation theory that are needed to prove the result on flattenings Theorem \ref{thm_edges} and we derive invariants for phylogenetic networks.

\section*{Acknowledgments}
{We would like to thank Mateusz Michalek for comments that improved a preliminary version of this manuscript (which allowed fixing a gap in one of the proofs) and we would like to thank the anonymous reviewers for fruitful comments that led to major improvements.} Both authors were partially supported
by the project reference PID2023-146936NB-I00 financed by the Spanish State Agency MCIN/AEI/10.13039/501100011033/ FEDER, UE, by the Severo Ochoa and Mar\'{\i}a de Maeztu Program for Centers and Units of Excellence in R\&D (project CEX2020-001084-M), and by the AGAUR project 2021 SGR 00603 Geometry of Manifolds and Applications, GEOMVAP.

\section{Preliminaries}

Throughout this section we describe the main notation used in the paper. The concepts related to phylogenetic trees and Markov processes on trees can be found in the book \cite{SteelPhylogeny}.
%

Given a tree $T$, we write
$V(T)$ and $E(T)$ for the set of nodes and edges of $T$, respectively. 
The set $V(T)$ splits into the set of leaves $L(T)$ (nodes of degree one) and the set of interior nodes $\Int(T):
V(T) = L(T) \cup \Int(T)$. One says that a tree is \emph{trivalent} if each node in $\Int(T)$ has degree 3. 
%

A \emph{phylogenetic tree} is a tree $T$ without nodes of degree 2 (so that each interior node represents a speciation event), together with a bijection between its leaves and a finite set $L$ representing biological entities. We denote by $n$ the cardinality of $L$.
A tree $T$ is \emph{rooted} if it has a distinguished node $r$, called the \emph{root}, which induces an orientation on the edges of $T$.

\paragraph{Markov processes on trees} Let $\Sigma$ be a finite set of cardinality $\kappa$ which represents the alphabet of possible states. For instance, $\Sigma=\{\tt A,C,G,T\}$ represents the set of four nucleotides adenine, cytosine, guanine and thymine. On a tree $T$, we consider random variables at its nodes taking values in $\Sigma$ .

{Given a rooted phylogenetic tree $T$, a \emph{Markov process} on $T$ is given by a \emph{distribution} $\pi=(\pi_x)_{x\in \Sigma}$ of states at the root (i.e. $\pi_x\geq 0$ and $\sum_{x\in \Sigma} \pi_x=1$) and \emph{row-stochastic} matrices ${M^e}$ (i.e. $M_{i,j}^e\geq 0$ and for each $i$, $\sum_j M_{i,j}^e=1$)  assigned to each oriented edge  $e: u\rightarrow v$, where $u,v$ are adjacent nodes. The Markov assumption means that descendant variables of each node are independent of the non-descendant given the parent variable (see \cite{ASCB2005,Sullivant2018}).}
{The \emph{hidden} Markov process on $T$ is specified by a {polynomial} map $\phi_T$ which maps each collection of parameters $(\pi, ({M^e})_{e\in E(T)})$ to the joint distribution of characters at the leaves of $T$, $p^T=(p_{x_1,\ldots,x_n})_{x_1,\ldots,x_n\in \Sigma}$.
}
%

\begin{ex}[{Hidden} Markov process on the tripod]\label{ex_tripod}\rm Consider the tree $T$ with set of leaves $L=\{a,b,c\}$ as in Figure \ref{figure:tripod}. This tree is called a \emph{tripod} and a Markov process on it is specified by a distribution $\pi$ at the internal node (which plays the role of the root $r$) and by transition matrices $A,B,C$ at the directed edges from $r$ to the leaves.  Then
the components of $p^T=\phi_T(\pi,A,B,C)$ are
\begin{eqnarray}\label{eq:tripod}
	p_{x,y,z}=\sum_{i\in \Sigma} \pi_i \, A_{i,x} \, B_{i,y} \, C_{i,z} \qquad x,y,z\in \Sigma.
\end{eqnarray}
\end{ex}


\begin{figure}
\begin{center}
 \includegraphics[scale=0.5]{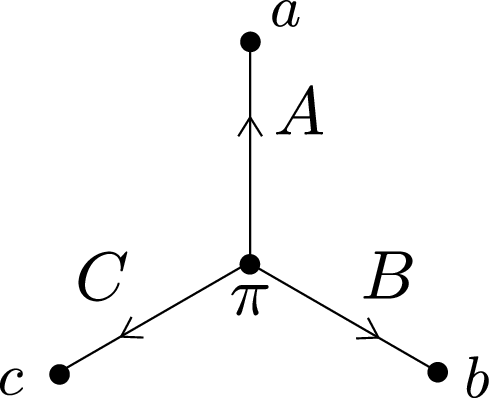}
 \caption{\label{figure:tripod} $T$ is the tripod with leaves labeled by $a,b,c$. A root distribution $\pi$ and transition matrices $A,B,C$ are attached to the root (interior node) and the edges, respectively. }
 \end{center}
\end{figure}

{Although in the probabilistic setting all vectors considered should represent distributions and hence should sum to one and be non-negative, 
it is useful to extend the setting to the complex field.
We will denote by $W$ the $\CC$-vector space $\CC^{\kappa}$ where we identify the standard basis with $\Sigma$, $W=\langle \Sigma \rangle_{\CC}$. This basis allows us to identify $W$ with its dual space $W^*$, and $\Hom(W,W)$ with $\kappa\times \kappa$-matrices (these identifications will be used throughout without further comment).
A \emph{normalized vector} in $W$ 
is a vector whose coordinates in this basis sum to one. The set of normalized vectors in $W$ will be denoted as $\mathcal{W}$.
In the same way, a \emph{normalized matrix} will be a $\kappa\times \kappa$ matrix whose rows sum to one and
$\MM$ will denote the set of all normalized $\kappa\times \kappa$ matrices.
%

We extend the polynomial map considered above and define $Par(T)=\mathcal{W} \times \prod_{e\in E(T)}\MM$ as the parameter space.
%
Then, the target space of the map $\phi_T$ is $\CC^{\kappa^n}$, which can be identified with $\LL:=\otimes^ n W$ (via the natural basis of $\otimes^n W$ given by $\Sigma$):
\begin{eqnarray*}
	\phi_T :  Par(T)  & \longrightarrow & \LL\\
	\big (\pi, \{{M^e}\}_{e\in E(T)} \big ) &  \mapsto & p^T = \sum_{x_1,\dots,x_n\in \Sigma} p_{x_1,\dots,x_n} x_1 \otimes \dots \otimes x_n.
\end{eqnarray*}
}

We shall write $V_T$ for the Zariski closure of this image, that is, the smallest algebraic variety containing it: $V_T=\overline{\Im \phi_T}$. Note that, as we are restricting to normalized parameters, $V_T$ is contained in the hyperplane defined by the  equation $\sum_{x_1,\dots,x_n}p_{x_1,\dots,x_n}=1.$

%
%
%

Parameters $(\pi,\{{M^e}\}_{e\in E(T)})$ are called \emph{non-singular} if $\pi_i\neq 0$ for all $i$ and $\det({M^e})\neq 0$ for all $e\in E(T)$.

The elements in the ideal $I(V_T)\subset R:=\CC[\{ p_{x_1,\ldots,x_n} \mid x_1,\ldots,x_n\in \Sigma \}]$ are referred to  as \emph{phylogenetic invariants}. Elements in $I(V_T)$ that lie in all $I(V_{T'})$ for any other phylogenetic tree $T'$ with leaf set $L$ are known as \emph{model invariants}. Phylogenetic invariants that are not model invariants are called \emph{topology invariants}.

\paragraph{Flattenings}

A  bipartition $A|B$ of $L$ is a decomposition of $L$ into two disjoint sets: $L=A \cup B$.
It naturally induces an isomorphism from the space of tensors to the space of 
${\kappa^{|A|}\times \kappa^{|B|}}$ matrices 
\begin{equation}
\begin{array}{rcl}
  \LL = \left(\otimes_{a \in A}W\right)\otimes \left(\otimes_{b \in B}W\right) & \longrightarrow & {Mat}_{\kappa^{|A|}\times \kappa^{|B|}} (\CC)\\
    p=(p_{x_1\dots x_n}) & \mapsto & flat_{A|B}(p)
\end{array},
\end{equation}
defined as follows: if $x_A=(x_l)_{l\in A}$ and $x_B=(x_l)_{l\in B}$, the $(x_A,x_B)$ entry of $flat_{A|B}(p)$ is $p_{x_A,x_B}$
For example, if we take {$\Sigma=\{\tt A,C,G,T\}$,} $L=\{1,2,3,4\}$ and $A=\{1,2\}$, $B=\{3,4\}$, we have

\[flat_{12|34}(p)=\left(\begin{array}{ccccc}
p_{\tt AAAA} & p_{\tt AAAC} & p_{\tt AA{AG}} &  \dots & p_{\tt AATT} \\
p_{\tt { AC}AA} & p_{\tt {AC}AC} & p_{\tt {AC}{ AG}} & \dots & p_{\tt { AC}TT}\\
p_{\tt AGAA} & p_{\tt AGAC} & p_{\tt AG{AG}} & \dots & p_{\tt AGTT} \\
\vdots & \vdots & \vdots & \vdots & \vdots\\
p_{\tt TTAA} & p_{\tt TTAC} & p_{\tt TT{AG}} & \dots & p_{\tt TTTT}
\end{array}\right).\]
If $T$ is a phylogenetic tree with $L(T)=L$, a bipartition $A|B$ of $L$ is an \emph{edge split} if it can be obtained by removing one of the non-pendant edges of $T$. Thanks to the following result, flattenings provide topology invariants of $V_T$:

\begin{thm}[\cite{Allman2008,CF11,Snyman}]\label{thm_AR}
Let $T$ be a phylogenetic tree and let $A|B$ be a bipartition of its set of leaves $L$. Let $p^T$ be a tensor obtained from a hidden Markov process on $T$. Then, if $A|B$ is an edge split on $T$,  $flat_{A|B}(p^T)$ has rank less than or equal to $\kappa$. Moreover, if $A|B$ is not an edge split and the parameters that generated $p^T$ are non-singular, rank of $flat_{A|B}(p^T)$ is larger than $\kappa$. In particular, for any edge split $A|B$ of $T$, the $(\kappa+1)\times (\kappa+1)$ minors of $flat_{A|B}(p)$ are topology invariants for $T$.
\end{thm}

{Buneman's theorem states that the topology of a tree can be recovered from the collection of its edge splits. 
This allows us to identify trivalent 4-leaf trees (quartets) with the split induced and write $T=A|B$ (see Figure \ref{figure:T1234}).}

\paragraph{$G$-equivariant models}
Several substitution models used in phylogenetics can be described in a very elegant way by the action of a permutation group acting on the set $\{\texttt{A,C,G,T}\}$ (see \cite{Draisma}). We adopt this approach and given the alphabet $\Sigma$, we consider a permutation group $G\leq \mathfrak{S}_{\kappa}$ acting on the standard basis of $W$ (identified with $\Sigma$).
The \emph{permutation representation} of $G$ on $W$ (or $\LL$) is the representation induced by extending linearly this action to all vectors in $W$.
We denote by $W^G$ and $\mathcal{W}^G$ the subspace of vectors in $W$ and $\mathcal{W}$ that remain invariant. Similarly, we denote by $\LL^G$ the subspace of $\LL$ composed of all the tensors invariant by the natural extension of the action of $G$ to $\LL$, that is, $G$ acts simultaneously on each factor $W$ using the original action on $W$: 
{\[\LL^G=\{p\in \otimes^n W \mid g\cdot p =p \, \, \forall g \in G\}.\]}
We write $\MM^G$ for the space of $G$-equivariant matrices in $\MM$: normalized matrices that remain invariant when permuting rows and columns according to the permutations ${g\in G}$ (${M}_{g(i),g(j)}={M}_{i,j}$ for any $g\in G$, and any $i,j\in \Sigma$). Equivalently, $K_g\, M\, K_g^{-1}=M$, where $K_g$ denotes the permutation matrix obtained by applying $g$ to the columns of $\mathrm{Id}$:
\begin{eqnarray*}
	\left(K_g\right)_{i,j}=
	\left\{ \begin{array}{rl}
		1 & \mbox{ if }j=g(i)\\
		0 & \mbox{otherwise.}
	\end{array}
	\right.
\end{eqnarray*}
\begin{rmk}\label{Kg_inverse}It is clear that $K_g$ is an orthogonal matrix and $K_g K_{g^{-1}}=K_{g^{-1}} K_g=\mathrm{Id}$. 
In particular, $K_g^{-1}=K_{g^{-1}}=K_g^t$. On the other hand, the reader may easily check that $\MM^G$ is multiplicatively closed: if $A,B\in \MM^G$, then $A\, B \in \MM^G$. Moreover, if $A\in \MM^G$ is invertible, then $A^{-1}\in \MM^G$.
\end{rmk} 



We define the \emph{$G$-equivariant substitution model} by taking $G$-invariant normalized vectors as root distributions and $G$-equivariant matrices as transition matrices. When $G$ is the trivial group formed by the neutral element, this model coincides with the one presented above and is known as the \emph{general Markov model}. 

\begin{ex}\label{ex:equivarmodels}\rm We present some of the well-known equivariant models on the set of states  $\Sigma=\{\texttt{A,C,G,T}\}$, arising from biological studies. We consider subgroups $G$ of the permutation group $\mathfrak{S}_{4}$ as follows.

\begin{itemize}
    \item[(a)]
    Consider the permutations $g_1=(\a,\c)(\g,\t)$ and $g_2=(\a,\g)(\c,\t)\in \mathfrak{S}_4$ and consider the permutation group $G =\langle g_1, g_2\rangle$. The $G$-invariant vectors of $W$ form the subspace spanned by $\textbf{1}$ and the equivariant matrices for the corresponding  $G$-equivariant model have the following structure
\begin{eqnarray} \label{K3_matrix}
 \begin{pmatrix}
  a & b & c & d \\
  b & a & d & c \\
  c & d & a & b \\
  d & c & b & a
 \end{pmatrix}.
\end{eqnarray}
 This $G$-equivariant model corresponds to the Kimura 3-parameter model (K81 briefly) introduced in \cite{Kimura1981}, and was motivated by the observation of different patterns of substitution among groups of nucleotides. 
 \item[(b)] If we impose $b=d$ in the K81 transition matrices \eqref{K3_matrix}, then we obtain the Kimura 2-parameter model (K80 briefly) that was proposed in \cite{Kimura1980}. This is an equivariant model with group generated by the permutations $g_1$ and $g_2$ in (a), together with $h=(\c,\t)$. Thus, the K80 group is isomorphic to the dihedreal group $D_8$. 
 \item[(c)] When one imposes $b=c=d$ in the K81 model, we obtain  the Jukes-Cantor model introduced in \cite{JC69}, JC69, which is the most simple model of nucleotide substitution. This is a $G$-equivariant model with $G= \mathfrak{S}_{4}$. 
 \end{itemize}
 We will work over these three models in Section 3.  In each case, the group $G$ will be renamed by the corresponding abbreviation so that $K81=\langle g_1,g_2\rangle$, $K80=\langle g_1,g_2,h\rangle$, and $JC69=\mathfrak{S}_{4}.$
\end{ex}

Given a rooted tree, the set of parameters of the corresponding $G$-equivariant model is $Par_G(T)= \mathcal{W}^G \times \prod_{e\in E(T)} \MM^G$.
The corresponding parameterization map is
\begin{eqnarray*}
	\phi^G_T :  Par_G(T)  & \longrightarrow & \LL
\end{eqnarray*}
and we denote by $V^G_T$ the Zariski closure of the image, that is, $V^G_T= \overline{\Im \phi^G_T}\subset \mathcal{L}$. {For any tree $T$, $V_T^G$ lies in $\LL^G$ (see \cite{Draisma})}.  {The defining ideal of $V_T^G \subset \LL$ will be denoted as $I_T^G$; note that as $V_T^G$ is an irreducible variety {(because $\phi_T^G$ is a continuous map from an irreducible variety),} $I_T^G\subseteq R$ is a prime ideal.
It is immediate to see that if $G_1\leq G_2$, then $V_T^{G_1}\supseteq V_T^{G_2}.$}

{Since the subspace $\LL^G$ is independent of any particular tree}, the (linear) equations defining $\LL^G$ within $\LL$ are model invariants. The dimension of the space $\LL^G$ for some particular permutation groups $G$ was given in \cite{CFK}, where it was proven that this space { of $G$-invariant tensors coincides with the space of mixtures of distributions arising from the $G$-equivariant model on phylogenetic trees.}

{Although $V_T^G \subset V_T \cap \LL^G$, the dimension of $V_T^G$ (which can be found in \cite{CFM}) is much larger than the dimension of $V_T$ minus the codimension of $\LL^G$, so a simple dimension count does not give any clue on whether $V_T\cap \LL^G$ coincides with $V_T^G$ or not.}

\begin{rmk}\label{change_root_equivariance}
{Although the parameterization $\phi_T$ depends on the root position $r$, the same joint distribution can be obtained with another root position if the parameters are changed conveniently (see, for instance, \cite{SHP97} or Proposition 1 of \cite{Allman2003}).}
{It is straightforward to check that the modifications of the parameters
preserve the $G$-invariance of the root distribution and the $G$-equivariance of the transition matrices.}
\end{rmk}

\paragraph{Notation} Some notation that will be used throughout the paper is the following. {Vectors are understood as column vectors and, given a vector $w\in \CC^{\kappa}$, we  denote by  $\mathrm{diag}(w)$ the diagonal matrix whose diagonal entries are the coordinates of $w$. Given a  $\kappa\times \kappa$ matrix $M$ and $y\in \Sigma$, $M_y$ denotes the $y$-th column of $M$. %
For notational convenience, we will write $M^{-t}$ for the transpose of the inverse of $M$, i.e. $(M^{-1})^t$. This will be used in Section 4 without further comment.  
We write $\mathbf{1}$ for the vector of ones, $\mathbf{1}=(1,1,\ldots,1)^t$.

\section{Motivating examples}\label{sec_ex}

In this section we proceed to show some examples trying to answer  Question 1 and motivating the results of the forthcoming sections by explaining how to get phylogenetic invariants in a simple way. In all cases we work with $\Sigma=\{\texttt{A,C,G,T}\}$ and identify these elements with the standard basis of $W$. The standard basis of $\LL=\otimes^n W$ is given naturally by tensor products of this basis and is ordered in lexicographical order. We work with the models introduced in Ex. \ref{ex:equivarmodels} and we do not require previous knowledge of representation theory: we only mention some connections to this theory and full details will be given in section 5.

\subsection{Kimura 3-parameter model (K81)}\label{sec:k81}

For the K81 model, the algebraic variety is usually described in Fourier coordinates {with respect to $G$}: if $p$ is a tensor in $\LL$ (understood as a column vector in the coordinates in the standard basis of $\LL$), consider the matrix
\[H=\begin{pmatrix} 1&1&1&1\\
1&1&-1&-1\\
1&-1&1&-1\\
1&-1&-1&1
\end{pmatrix}\]
and perform the change of coordinates $\bar{p}=(H^{-1}\otimes \dots\otimes H^{-1})\, p$, where $H^{-1}=\frac{1}{4}H$. It is well established that, using these coordinates, the ideal of the phylogenetic variety is a binomial ideal (see \cite{Evans1993}, \cite{Sturmfels2005}).
The basis $B$ of $W$ associated to these coordinates is induced by the columns of $H$,
\begin{eqnarray}
\label{FourierBasis}    
\bar{\a} =&  \a+\c+\g+\t =& (1,1,1,1)\\
\bar{\c} =& \a+\c-\g-\t =& (1,1,-1,-1) \nonumber\\
\bar{\g}  = & \a-\c+\g-\t=& (1,-1,1,-1) \nonumber\\
\bar{\t}  = &\a-\c-\g+\t =& (1,-1,-1,1) \nonumber
\end{eqnarray}

Let us see what the constraints of $\LL^{K81}$ impose on $\bar{p} \in \LL$. If we apply the action of $g_1$ and $g_2$ to the vectors in {$B=\{\bar{\a},\bar{\c},\bar{\g},\bar{\t}\}$} we get
\begin{eqnarray}\label{action_g}
\begin{array}{lllllll}
g_1\, \ba =\ba, & \quad & g_1\,\bc=\bc, & \quad  & g_1\,\bg=-\bg, & \quad  & g_1\,\bt =-\bt, \\
g_2\,\ba =\ba, & \quad & g_2\,\bc=-\bc, & \quad & g_2\,\bg=\bg, & \quad & g_2\,\bt =-\bt  \,.
\end{array}
\end{eqnarray}
By identifying $K81$ with the additive group $\ZZ/2\ZZ \times \ZZ/2\ZZ$ via $\a \leftrightarrow (0,0)$, $\c \leftrightarrow (0,1)$, $\g \leftrightarrow (1,0)$, $\t \leftrightarrow (1,1)$ we have the following result:

\begin{lema}\label{modeq:K81}
(Model invariants for $K81$)
A tensor $p \in \LL$ is invariant by the action of the group $K81$ (i.e. belongs to $\LL^{K81}$) if and only if  $\bar{p}_{x_1\dots x_n}=0$ whenever $x_1+\dots+x_n\neq (0,0)\in \ZZ/2\ZZ \times \ZZ/2\ZZ.$
\end{lema}

\begin{proof}
Note that the set of tensors
\[\Omega=\{p \in \LL \mid \bar{p}_{x_1\dots x_n}=0 \textrm{ if } x_1+\dots+x_n\neq (0,0)\}\]
is invariant by $K81$. Indeed, note that $x_1+\dots+x_n \equiv (a,b)$, where $a$ is the number of $\g$ plus the number of $\t$ among the $x_i$ mod 2, and $b$ is the number of $\c$  plus the number of $\t$ mod 2. Thus, $x_1+\dots+x_n =  (0,0)$ if and only if $\sharp \g+\sharp \t \equiv 0$ and $\sharp \c+\sharp \t \equiv 0$ in $\ZZ/2\ZZ$ (i.e. $x_1\dots x_n$ has the same parity of $\c$'s, $\g$'s and $\t$'s). By (\ref{action_g}) this holds if and only if the action of $g_1$ and $g_2$ leaves the coordinate $\bar{p}_{x_1\dots x_n}$ invariant.

Now the lemma follows by dimension count:  $\LL^{K81}$ has dimension $4^{n-1}$ (see \cite{CFK}), which coincides with the dimension of $\Omega$.
\end{proof}

This lemma was known in the previous literature for tensors $p$ in the image of $\phi_T$ for some tree $T$. From this result we obtain a system of linear equations defining $\LL^{K81}$, which is the {space} of mixtures of distributions on trees on $n$ leaves (see   \cite{CFK}). Now let us look at the equations of $V_T$ coming from flattenings and see how they add to these model invariants.

We consider {$n=4$ and $L=\{1,2,3,4\}$. Then
\begin{equation}\label{eq:flat}
flat_{12|34}(\bar{p})=(H^{-1}\otimes H^{-1})\, flat_{12|34}(p) \, (H^{-t}\otimes H^{-t}).
\end{equation}
}

\begin{figure}
\begin{center}
 \includegraphics[scale=0.5]{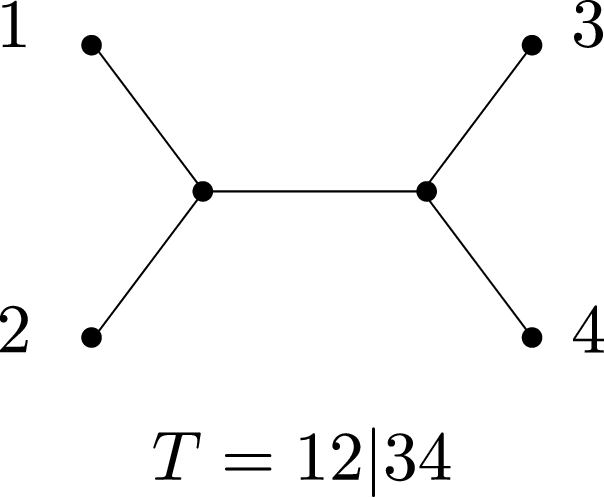}
 \caption{\label{figure:T1234} The trivalent quartet tree $T$ with edge split $12|34$. }
 \end{center}
\end{figure}

According to the previous lemma, $flat_{12|34}(\bar{p})$ can be written as a block diagonal matrix if we choose the following order on rows and columns, $\texttt{AA,CC,GG,TT,AC,CA,GT}$, $\texttt{TG}$, $\texttt{AG,CT,GA,TC,AT,CG,GC,TA}$:

\begin{eqnarray}
\label{flat_K3}
flat_{12|34}(\bar{p})=\begin{pmatrix}
    B_{\texttt A} &  &  &  \\
     & B_{\texttt C} &  &  \\
     &  & B_{\texttt G} & \\
     & &  & B_{\texttt T}
  \end{pmatrix}
  \end{eqnarray}
where
$$B_{\texttt A}=
\begin{pmatrix}
   \bar{p}_{\texttt{AAAA}}& \bar{p}_{\texttt{AACC}} & \bar{p}_{\texttt{AAGG}} & \bar{p}_{\texttt{AATT}} \\
  \bar{p}_{\texttt{CCAA}} & \bar{p}_{\texttt{CCCC}} & \bar{p}_{\texttt{CCGG}}& \bar{p}_{\texttt{CCTT}} \\
  \bar{p}_{\texttt{GGAA}} & \bar{p}_{\texttt{GGCC}} & \bar{p}_{\texttt{GGGG}} & \bar{p}_{\texttt{GGTT}} \\
  \bar{p}_{\texttt{TTAA}} & \bar{p}_{\texttt{TTCC}} & \bar{p}_{\texttt{TTGG}} & \bar{p}_{\texttt{TTTT}}
\end{pmatrix}, \, B_{\texttt C}=  \begin{pmatrix}
   \bar{p}_{\texttt{ACAC}}& \bar{p}_{\texttt{ACCA}} & \bar{p}_{\texttt{ACGT}} & \bar{p}_{\texttt{ACTG}} \\
    \bar{p}_{\texttt{CAAC}}& \bar{p}_{\texttt{CACA}} & \bar{p}_{\texttt{CAGT}} & \bar{p}_{\texttt{CATG}} \\
      \bar{p}_{\texttt{GTAC}}& \bar{p}_{\texttt{GTCA}} & \bar{p}_{\texttt{GTGT}} & \bar{p}_{\texttt{GTTG}} \\
        \bar{p}_{\texttt{TGAC}}& \bar{p}_{\texttt{TGCA}} & \bar{p}_{\texttt{TGGT}} & \bar{p}_{\texttt{TGTG}}
\end{pmatrix}, $$
$$B_{\texttt G} = \begin{pmatrix}
   \bar{p}_{\texttt{AGAG}}& \bar{p}_{\texttt{AGCT}} & \bar{p}_{\texttt{AGGA}} & \bar{p}_{\texttt{AGTC}} \\
       \bar{p}_{\texttt{CTAG}}& \bar{p}_{\texttt{CTCT}} & \bar{p}_{\texttt{CTGA}} & \bar{p}_{\texttt{CTTC}} \\
     \bar{p}_{\texttt{GAAG}}& \bar{p}_{\texttt{GACT}} & \bar{p}_{\texttt{GAGA}} & \bar{p}_{\texttt{GATC}} \\
       \bar{p}_{\texttt{TCAG}}& \bar{p}_{\texttt{TCCT}} & \bar{p}_{\texttt{TCGA}} & \bar{p}_{\texttt{TCTC}}
\end{pmatrix}, \,
B_{\texttt T} = \begin{pmatrix}
   \bar{p}_{\texttt{ATAT}}& \bar{p}_{\texttt{ATCG}} & \bar{p}_{\texttt{ATGC}} & \bar{p}_{\texttt{ATTA}} \\
   \bar{p}_{\texttt{CGAT}}& \bar{p}_{\texttt{CGCG}} & \bar{p}_{\texttt{CGGC}} & \bar{p}_{\texttt{CGTA}} \\
  \bar{p}_{\texttt{GCAT}}& \bar{p}_{\texttt{GCCG}} & \bar{p}_{\texttt{GCGC}} & \bar{p}_{\texttt{GCTA}} \\
  \bar{p}_{\texttt{TAAT}}& \bar{p}_{\texttt{TACG}} & \bar{p}_{\texttt{TAGC}} & \bar{p}_{\texttt{TATA}}
\end{pmatrix},
$$
{(the subindices of these four blocks coincide with the sum of the first two indices of the coordinates).}

By Theorem \ref{thm_AR} {and \eqref{eq:flat},} if $T=12|34$ as in Figure \ref{figure:T1234}, this matrix has rank $\leq 4$ for any $p\in V_T$. Consider tensors $\bar{p}$ in the open set
\begin{equation}\label{eq:O}
    \mathcal{O}=\{\bar{p}_{\texttt{AAAA}}\neq 0, \bar{p}_{\texttt{ACAC}}\neq 0, \bar{p}_{\texttt{AGAG}}\neq 0, \bar{p}_{\texttt{ATAT}}\neq 0\}.
\end{equation}

Then, as there is an element in each block which is different from zero, $flat_{12|34}(\bar{p})$ has rank $\leq 4$ if and only if each block has rank 1. In other words, by observing that this open set meets $V_T^{K81}$ properly (and hence defines a dense subset), we recover the following well known result.
\begin{lema}\label{lem_2minors}
The $2\times 2$ minors of each block $B_{\texttt A}$, $B_{\texttt C}$, $B_{\texttt G}$ and $B_{\texttt T}$ are phylogenetic invariants for the tree $T=12|34$.
\end{lema}
Moreover, these are \emph{topology} invariants due to Theorem \ref{thm_AR}. These equations were first obtained in \cite{Sturmfels2005} by using Fourier coordinates and can be obtained independently by the  tools of representation theory for $G$-equivariant models developed in \cite{Draisma} (see \cite{CF11}). Note that we have obtained this result in a direct way from Theorem \ref{thm_AR} by adjoining  the constraints of $\LL^{K81}$.  In section 5 we prove that both approaches are equivalent for any $G$-equivariant model, so the simple  way of getting these equations as explained above can be reproduced for all models and trees. Moreover in  Example \ref{sec5_ExK81} we give an interpretation of the open set $\mathcal{O}$ in terms of marginalizations of the tensor.


\begin{rmk}
The Fourier basis introduced above is consistent with the Maschke decomposition of $W$ into the isotypic components induced by the permutation representation of the group $K81$ on $W$ (see \cite[Example 5.3]{CF11} {and Example \ref{ex_K80} below).} Similarly, the basis
\[B^n=\{\bar{x}_1\otimes\dots\otimes \bar{x}_n \mid\, \bar{x}_i\in B\}\]
of $\otimes^n W$ is adapted to the isotypic components of $\LL$.
\end{rmk}

\subsection{Kimura 2-parameter model (K80)}\label{sec:k80}

In this subsection we derive equations for the K80 model just by imposing the constraints of $\LL^{K81}$ to the equations obtained for K81 above.
%

%

\subsubsection{The K80 model on tripods}
We first study tripod trees and obtain the following result that gives a positive answer to Question 1.
\begin{lema}\label{lem_tripodk80}
    If $T$ is the tripod tree, the intersection $V_T^{K81}\cap \LL^{K80}$ is an irreducible variety which coincides with $V_T^{K80}$. Moreover, we have the following equality in terms of ideals of the ring $R=\mathbb{C}[\{p_{x_1\dots x_n}\mid x_1,\dots,x_n\in \Sigma\}]$: $I(V_T^{K81})+I(\LL^{K80})=I(V_T^{K80})$.
\end{lema}
\begin{proof}
   In Appendix \ref{sec:app1} we prove the equality of ideals by using Macaulay2 \cite{M2} and the computation done in Small Phylogenetic trees webpage \url{https://www.coloradocollege.edu/aapps/ldg/small-trees/} (see \cite{Smalltrees}).
   From this, the intersection $V_T^{K81} \cap \mathcal{L}^{K80}$ is equal to $V_T^{K80}$ and, in particular, it is irreducible.
   \end{proof}
Actually, the answer to Question 1 would require working with the variety $V_T$ of the general Markov model on $T$ instead of $V_T^{K81}$. However, we will see in Corollary \ref{cor_subgrup} that one can work on the intersection from submodels. It is worth noting that even for the tripod, a generating set for the ideal $I(V_T)$ remains unknown (the problem of giving a set of generators is known as the Salmon conjecture \cite[Conjecture 3.24]{ASCB2005}), so we could not have performed the computations from the general Markov model directly.

\subsubsection{K80 model on quartets}\label{sec_quartetsK80}

\begin{rmk} (Equations defining $\LL^{K80}$ when $n=4$)
\label{modeq:K80}
%
The space $\LL^{K80}$ is defined by the equations in Lemma \ref{modeq:K81} {together with new constraints $h\cdot p=p$
(because $h$ together with the trivial permutation $id$ form a transversal of $K81\backslash K80$).
Translating this into Fourier coordinates, we get that} the
model invariants for $K81$ displayed in Lemma \ref{modeq:K81} together with the following 28 linear equations 
form a minimal set of equations defining $\LL^{K80}$ within $\LL$:
\[
\begin{array}{cccc}
\bar{p}_{\a\a\c\c} = \bar{p}_{\a\a\t\t} & \bar{p}_{\c\c\a\a} = \bar{p}_{\t\t\a\a} & \bar{p}_{\c\c\c\c} = \bar{p}_{\t\t\t\t} & \bar{p}_{\c\c\g\g} = \bar{p}_{\t\t\g\g} \\ 
\bar{p}_{\c\c\t\t} = \bar{p}_{\t\t\c\c} & \bar{p}_{\g\g\c\c} = \bar{p}_{\g\g\t\t} & \bar{p}_{\g\a\c\t} = \bar{p}_{\g\a\t\c} & \bar{p}_{\a\g\c\t} = \bar{p}_{\a\g\t\c} \\
\bar{p}_{\c\t\a\g} = \bar{p}_{\t\c\a\g} & \bar{p}_{\c\t\c\t} = \bar{p}_{\t\c\t\c} & \bar{p}_{\c\t\g\a} = \bar{p}_{\t\c\g\a} & \bar{p}_{\c\t\t\c} = \bar{p}_{\t\c\c\t}\\
\bar{p}_{\a\c\c\a} = \bar{p}_{\a\t\t\a} & \bar{p}_{\a\c\g\t} = \bar{p}_{\a\t\g\c} & \bar{p}_{\a\c\t\g} = \bar{p}_{\a\t\c\g} & \bar{p}_{\c\a\a\c} = \bar{p}_{\t\a\a\t}\\
\bar{p}_{\c\a\c\a} = \bar{p}_{\t\a\t\a} & \bar{p}_{\c\a\g\t} = \bar{p}_{\t\a\g\c} & \bar{p}_{\c\a\t\g} = \bar{p}_{\t\a\c\g} & \bar{p}_{\g\c\a\t} = \bar{p}_{\g\t\a\c}\\
\bar{p}_{\g\c\t\a} = \bar{p}_{\g\t\c\a} & \bar{p}_{\g\c\g\c} = \bar{p}_{\g\t\g\t} & \bar{p}_{\g\c\c\g} = \bar{p}_{\g\t\t\g} & \bar{p}_{\c\g\a\t} = \bar{p}_{\t\g\a\c}\\
\bar{p}_{\c\g\t\a} = \bar{p}_{\t\g\c\a} & \bar{p}_{\c\g\g\c} = \bar{p}_{\t\g\g\t} & \bar{p}_{\c\g\c\g} = \bar{p}_{\t\g\t\g} & 
\bar{p}_{\a\c\a\c} = \bar{p}_{\a\t\a\t}.
\end{array}\]
Note that the codimension of $\LL^{K80}$ within $\LL^{K81}$ is indeed 28 (see \cite[Prop. 20]{CFK}).
\end{rmk}
%

Now we consider the quartet tree $T=12|34$ as done in section \ref{sec:k81}. If $p\in V_T^{K80}$, the flattening $flat_{12|34}(\bp)$ in (\ref{flat_K3}) has four blocks again, which must have rank $\leq 1$ for tensors in $V_T^{K80}$ (as $V_T^{K80}\subset V_T^{K81})$.
From the equations in Remark \ref{modeq:K80} we get that {blocks $B_{\texttt C}$ and $B_{\texttt T}$ have the same entries up to rearrangements of rows and columns,} so we only need to consider the $2\times 2$ minors of $B_{\texttt A}, B_{\texttt C}$, and $B_{\texttt G}$.
Moreover, Remark \ref{modeq:K80} gives identities between the entries of  block $B_{\texttt A}$. For example, the minor formed by the first two rows and the first and fourth columns becomes
$\bp_{\a\a\a\a}\,\bp_{\c\c\t\t}-\bp_{\a\a\c\c}\,\bp_{\c\c\a\a}=0$. Subtracting this from the first minor $\bp_{\a\a\a\a}\,\bp_{\c\c\c\c}-\bp_{\a\a\c\c}\,\bp_{\c\c\a\a}$ of $B_{\texttt A}$, we get $\bp_{\a\a\a\a}\,(\bp_{\c\c\c\c}-\bp_{\c\c\t\t})=0$. Thus, as the ideal $I(V_T^{K80})$ is prime and $\bp_{\a\a\a\a}$ does not vanish at all points of $V_T^{K80}$, we obtain that
%
\begin{eqnarray*}
\bp_{\c\c\c\c}-\bp_{\c\c\t\t}
\end{eqnarray*}
is a linear phylogenetic invariant for the model $K80$.
Similarly, working with $B_{\g}$, we obtain the phylogenetic linear invariant
\begin{eqnarray*}
\bp_{\c\t\c\t}-\bp_{\c\t\t\c}.
\end{eqnarray*}
These two linear invariants define the same linear variety as the invariants discovered by Lake in \cite{Lake1987}.
Besides these two, the rank constraints for the blocks of $flat_{12|34}(\bar{p})$ produce a total of  54 quadratic phylogenetic invariants (a set of non-redundant $2\times 2$ minors).

This particular example shows that in general Question 1 does not have a positive answer in terms of ideals: we have $I_T^{K80}\neq I_T^{K81}+I(\LL^{K80})$ for $T=12|34$. Indeed, $\bp_{\texttt{CCCC}}-\bp_{\texttt{CCTT}}$ lies in $I_T^{K80}$  but not in $J=I_T^{K81}+I(\LL^{K80})$ because the linear part of $J$ coincides with $I(\LL^{K80}).$ With Macaulay2 computations (and the help of the package Binomials \cite{kahle10}, see Appendix \ref{sec:app1}) we found that  $I_T^{K81}+I(\LL^{K80})$ has 93 minimal primes, which gives a negative answer to Question 1 in terms of varieties as well. One of the primes corresponds to $V_T^{K80}$ and there are 60 other primes of  degree 2 and 32 linear primes.


\subsection{Jukes-Cantor model (JC69)}\label{sec:jc69}

Now we proceed to obtain equations for JC69 from those obtained for K80 when imposing the constraints of $\LL^{JC69}$.

\subsubsection{JC69 model on tripods}
In this case we get an analogous result to the K80 case (see Appendix \ref{sec:app2} for a computational proof):
\begin{lema}\label{lem_tripodjc69}
    Let $T$ be the tripod tree. Then the intersection $V_T^{K80}\cap \LL^{JC69}$ is an irreducible variety which coincides with $V_T^{JC69}$. Moreover, we have the following equality in terms of ideals $I(V_T^{K80})+I(\LL^{JC69})=I(V_T^{JC69}).$
\end{lema}

\subsubsection{JC69 model on quartets}

\begin{rmk}
\label{modeq:JC}
(Equations defining $\LL^{JC69}$ when $n=4$) A minimal set of equations defining $\LL^{JC69}$ can be obtained from the  equations displayed in Remark \ref{modeq:K80} together with the following 21 equations:
\[\begin{array}{cccccc}
\bar{p}_{\a\a\g\g} = \bar{p}_{\a\a\t\t} & \bar{p}_{\c\c\g\g} = \bar{p}_{\c\c\t\t} & \bar{p}_{\g\g\a\a} = \bar{p}_{\t\t\a\a} & \bar{p}_{\g\g\c\c} = \bar{p}_{\t\t\c\c} &
\bar{p}_{\g\g\g\g} = \bar{p}_{\t\t\t\t} & \bar{p}_{\a\c\a\c} = \bar{p}_{\a\g\a\g} \\
\bar{p}_{\a\c\c\a} = \bar{p}_{\a\g\g\a} & \bar{p}_{\a\c\g\t} = \bar{p}_{\a\c\t\g} & \bar{p}_{\a\c\g\t} = \bar{p}_{\a\g\c\t} & \bar{p}_{\g\t\c\a} = \bar{p}_{\t\g\c\a} &
\bar{p}_{\g\t\g\t} = \bar{p}_{\t\g\t\g} & \bar{p}_{\c\a\a\c} = \bar{p}_{\g\a\a\g} \\
\bar{p}_{\c\a\c\a} = \bar{p}_{\g\a\g\a} & \bar{p}_{\c\a\g\t} = \bar{p}_{\c\a\t\g} & \bar{p}_{\c\a\g\t} = \bar{p}_{\g\a\c\t} & \bar{p}_{\t\c\g\a} = \bar{p}_{\t\g\c\a} &
\bar{p}_{\t\c\t\c} = \bar{p}_{\t\g\t\g} & \bar{p}_{\g\t\a\c} = \bar{p}_{\t\g\a\c}.
\end{array}\]
{This can be seen by using the transversal $\{id, m_1 = (\a\c), m_2 = (\a\t)\}$ of $K80\backslash JC69$ and translating the new constraints $m_1\cdot p=p$, $m_2\cdot p=p$ into Fourier coordinates.} 

Note that the codimension of $\LL^{JC69}$ within $\LL^{K80}$ is 21, in concordance with \cite[Prop. 20]{CFK}.
%
\end{rmk}

We consider again the tree $T=12|34$. To obtain phylogenetic invariants for the JC69 model on $T$, we add the constraints in Remark \ref{modeq:JC} to those already obtained for K80.  
%
Using them, the blocks $B_{\X}$ become
\[B_\a = \left ( \begin{array}{cccc}
\bar{p}_{\a\a\a\a} & \bar{p}_{\a\a\c\c} & \bar{p}_{\a\a\c\c} & \bar{p}_{\a\a\c\c} \\
\bar{p}_{\c\c\a\a} & \bar{p}_{\c\c\c\c} & \bar{p}_{\c\c\g\g} & \bar{p}_{\c\c\g\g} \\
\bar{p}_{\c\c\a\a} & \bar{p}_{\c\c\g\g} & \bar{p}_{\c\c\c\c} & \bar{p}_{\c\c\g\g} \\
\bar{p}_{\c\c\a\a} & \bar{p}_{\c\c\g\g} & \bar{p}_{\c\c\g\g} & \bar{p}_{\c\c\c\c}
\end{array} \right )
 \qquad
 B_{\g} = \left (
\begin{array}{cccc}
\bar{p}_{\a\c\a\c} & \bar{p}_{\a\c\g\t} & \bar{p}_{\a\c\c\a}  & \bar{p}_{\a\c\g\t}  \\
\bar{p}_{\c\g\a\t} & \bar{p}_{\c\g\c\g} & \bar{p}_{\c\g\t\a} & \bar{p}_{\c\g\g\c} \\
\bar{p}_{\c\a\a\c} & \bar{p}_{\c\a\g\t} & \bar{p}_{\c\a\c\a} & \bar{p}_{\c\a\g\t} \\
\bar{p}_{\c\g\a\t} & \bar{p}_{\c\g\g\c} & \bar{p}_{\c\g\t\a} & \bar{p}_{\c\g\c\g}
\end{array} \right ).
\]
%
\[ B_{\c} = \left (
\begin{array}{cccc}
\bar{p}_{\a\c\a\c} & \bar{p}_{\a\c\c\a} & \bar{p}_{\a\c\g\t} & \bar{p}_{\a\c\g\t}  \\
\bar{p}_{\c\a\a\c} & \bar{p}_{\c\a\c\a} & \bar{p}_{\c\a\g\t} & \bar{p}_{\c\a\g\t} \\
\bar{p}_{\c\g\a\t} & \bar{p}_{\c\g\t\a} & \bar{p}_{\c\g\c\g} & \bar{p}_{\c\g\g\c} \\
\bar{p}_{\c\g\a\t} & \bar{p}_{\c\g\t\a} & \bar{p}_{\c\g\g\c} & \bar{p}_{\c\g\c\g}
\end{array} \right ) \]
(block $B_{\t}$ is omitted as it contains the same entries as $B_{\c}$, up to rearrangements of rows and columns). 
As already noted in the K80 model, some rank equations obtained from these blocks now become redundant. This phenomenon can be also understood by making use of the representation theory of the groups involved (see forthcoming Example \ref{ex_K80}). It is to avoid this redundancy that in the forthcoming section 5 we invoke the concept of thin flattening introduced in \cite{CF11} rather the usual flattening of \cite{Draisma}.

%
Similar computations to those performed for the $K80$ model give rise to phylogenetic invariants: 2 linear invariants 
\begin{eqnarray*}
\bar{p}_{\c\c\g\g} = \bar{p}_{\c\c\c\c} \qquad
\bar{p}_{\c\g\c\g} = \bar{p}_{\c\g\g\c},
\end{eqnarray*}
plus 10 quadrics
\begin{eqnarray*}
\bar{p}_{\a\a\a\a}\,\bar{p}_{\c\c\c\c} - \bar{p}_{\a\a\c\c}\,\bar{p}_{\c\c\a\a} = 0 &
\bar{p}_{\a\c\a\c}\,\bar{p}_{\c\a\c\a} - \bar{p}_{\a\c\c\a}\,\bar{p}_{\c\a\a\c} = 0 \\
\bar{p}_{\a\c\a\c}\,\bar{p}_{\c\g\t\a} - \bar{p}_{\a\c\c\a}\,\bar{p}_{\c\g\a\t}= 0 &
\bar{p}_{\c\a\a\c}\,\bar{p}_{\c\g\t\a} - \bar{p}_{\c\a\c\a}\,\bar{p}_{\c\g\a\t} = 0 \\
\bar{p}_{\a\c\a\c}\,\bar{p}_{\c\a\g\t} - \bar{p}_{\a\c\g\t}\,\bar{p}_{\c\a\a\c} = 0 &
\bar{p}_{\a\c\a\c}\,\bar{p}_{\c\g\c\g} - \bar{p}_{\a\c\g\t}\,\bar{p}_{\c\g\a\t} = 0 \\
\bar{p}_{\c\a\a\c}\,\bar{p}_{\c\g\c\g} - \bar{p}_{\c\a\g\t}\,\bar{p}_{\c\g\a\t} = 0 &
\bar{p}_{\a\c\c\a}\,\bar{p}_{\c\a\g\t} - \bar{p}_{\a\c\g\t}\,\bar{p}_{\c\a\c\a} = 0 \\
\bar{p}_{\a\c\c\a}\,\bar{p}_{\c\g\c\g} - \bar{p}_{\a\c\g\t}\,\bar{p}_{\c\g\t\a} = 0 &
\bar{p}_{\c\a\c\a}\,\bar{p}_{\c\g\c\g} - \bar{p}_{\c\a\g\t}\,\bar{p}_{\c\g\t\a} = 0.
\end{eqnarray*}

\section{The main result}
The aim of this section is to prove the main result of the paper (Theorem \ref{thm_main}). We start by explaining the marginalization procedure that will be needed to apply induction.

We consider again a finite set $\Sigma$ with cardinality $\kappa$ and $W=\langle \Sigma \rangle_{\mathbb{C}}$. Given a vector $w\in W$, write $(w_1,\dots,w_{\kappa})$ for its components in the basis induced by $\Sigma$. In the space $\otimes^n W$ we can define the \emph{marginalization} over the last component as the map: 
\[\begin{array}{rcl}
f_n:\otimes^n W &\rightarrow & \otimes^{n-1}W\\
w^1\otimes\dots \otimes w^n &\mapsto & \left ( \sum_{i \in \Sigma}w^n_i\right ) \; (w^1\otimes\dots\otimes w^{n-1}) 
\end{array}\]
(and extended by linearity). 
For any $l\in L$, we introduce the notation $W_l$ to denote the copy of $W$ in $\otimes^n W$ corresponding to $l$. 
In the space $\bigotimes_{u\in L}W_u$ the \emph{marginalization} over the component $l\in L$ is defined accordingly as the map
\[\begin{array}{rcl}
f_l:\otimes^n W=\bigotimes_{u\in L}W_u &\rightarrow &\bigotimes_{u \neq l}W_u\\
\otimes_{u \in L}\; w^u &\mapsto & \left (\mathbf{1}^t\cdot w^l \right ) \;  \left(\otimes_{u \neq l}\; w^u\right)
\end{array}\]
{(noting that $\textbf{1}^t\cdot w^l=\sum_{i \in \Sigma}w^l_i$).} 
If $T$ is a tree, let $l$ be one of its leaves and $T'$ be the tree obtained from $T$ by \emph{pruning} $l$ {(that is, $T'$ is obtained by removing $l$ and the pendant edge adjacent to it, and suppressing the interior node adjacent to this edge).} 
Then the marginalization map satisfies  $f_l(\Im \varphi_T^G)=\Im \varphi_{T'}^G$ (see \cite[Lemma 4.11]{CFM}).



\begin{lema}\label{marginalisation}
For any $p\in \LL^G$ and $l\in L$, $f_l(p)$ is also $G$-invariant. 
\end{lema}


\begin{proof}
A tensor $p=\sum_{x_1,\ldots,x_n\in \Sigma} p_{x_1,\ldots,x_n} x_1 \otimes \ldots \otimes x_n$ is $G$-invariant if and only if $p_{x_1,\ldots,x_n} = p_{gx_1,\ldots,gx_n}$ for any $g\in G$.
	Without loss of generality, we may assume that the leaf $l$ is the last leaf of $T$.
	If $q=f_l(p)$, then $	q_{x_1,\ldots,x_{n-1}}=\sum_{s\in \Sigma} p_{x_1,\ldots,x_{n-1},s}
$ and for any $g\in G$,
\begin{eqnarray*}
	q_{\op{g}{x_1},\ldots,\op{g}{x_{n-1}}}=\sum_{s\in \Sigma} p_{\op{g}{x_1},\ldots,\op{g}{x_{n-1}},s}=\sum_{x=g^{-1}s\in \Sigma} p_{\op{g}{x_1},\ldots,\op{g}{x_{n-1}},\op{g}{x}}
\end{eqnarray*}
The claim follows trivially from here.
\end{proof}



For complex parameters, let $U\subset \mathcal{M}$ be the open subset of normalized $\kappa\times \kappa$ matrices defined as
\[U=\{M\in \mathcal{M} \mid \, {M}_{i,i} \neq {M}_{j,i} \; \mbox{ if } i\neq j\}.\]
If we work over the real field, this set includes an important class of matrices: a matrix $M$ is \emph{DLC} (for Diagonal Largest in Column) if ${M}_{i,i} > {M}_{j,i} $ for all $i\neq j$; the set of DLC matrices has played an important role in the phylogenetics literature as DLC transition matrices can be unequivocally identified from the distribution at the leaves of a tree, see \cite{chang1996}. 

\subsection{The tripod}

Consider a Markov process on the tripod tree $T$ of Figure \ref{figure:tripod} with leaves $a,b,c$, transition matrices $A$, $B$, $C$, and distribution $\pi$ at the root as in Example \ref{ex_tripod}.

\begin{prop}\label{thm_tripod} Let $T$ be the tripod tree and let $p=\phi_T(\pi;A,B,C)$ be the image of non-singular parameters.
If $p$ is $G$-invariant and one of the transition matrices lies in $U$, then $\pi$ is $G$-invariant and the matrices $A,B,C$ are $G$-equivariant.
\end{prop}

Our proof is inspired by \cite{chang1996}. Without loss of generality we can assume that $C$ is in $U$.
First we consider the image of $p$ by the marginalization map over leaf $c$, $f_{c}(p) \in W\otimes W$, and from it we define the matrix $J^{ab}$ as $J^{ab}_{i,j}=(f_c(p))_{i,j}=\sum_kp_{ijk}$.
Then, from \eqref{eq:tripod} we get
\begin{equation}\label{eq_J}
J^{ab}=A^t\, \mathrm{diag}(\pi)\, B.
\end{equation}



Given $s\in \Sigma$, write $P^{s}$ for the $\kappa\times \kappa$ matrix given as the slice of $p$ with fixed third coordinate $s$ (with rows labelled by the states in $a$ and columns labelled by the states in $b$), {that is, $P^{s}_{i,j}=p_{i,j,s}$, for $i,j\in \Sigma$.} 
As we have non-singular parameters, the matrix $J^{ab}$ is invertible and we can consider the matrix
\[Q^s =(J^{ab})^{-1}\, P^{s}.\]

We need the following lemma.
\begin{lema}\label{lem_propietats} If $p$ is $G$-invariant, then
\begin{itemize}
\item[(i)] $J^{ab}$ is $G$-equivariant.
\item[(ii)]\label{eq_pgamma}
$K_g\, P^{s}\, K_g^{-1} = P^{\op{g^{-1}}{s}}$, for all $g\in G$.
\item[(iii)] $K_g\, Q^s\, K_g^{-1} = Q^{\op{g^{-1}}{s}}  $ for all $g\in G$. In particular, $Q^{s}$ and $Q^{\op{g^{-1}}{s}}$ are similar matrices and share the same eigenvalues.
\end{itemize}
\end{lema}

\begin{proof}
\begin{itemize}
 \item[$(i)$] By Lemma \ref{marginalisation}, $f_c(p)$ is a $G$-invariant  tensor and hence the matrix $J^{ab}$ is $G$-equivariant.
 \item[$(ii)$] Note that $P^{s}_{\op{g}{i},\op{g}{j}}=p_{\op{g}{i},\op{g}{j},{s}}$, which is equal to $p_{i, j, g^{-1}\, {s}}$ because $p$ is $G$-invariant. Thus, $K_g P^{s} K_g^{-1} =P^{g^{-1} {s}}$.
 \item[$(iii)$] {Since $J^{ab}$ is $G$-equivariant, so is $(J^{ab})^{-1}$ (Remark \ref{Kg_inverse}).} Now, the claim follows directly from $(i)$ and $(ii)$ and the definition of $Q^s$.
 %
\end{itemize}
\end{proof}

We can proceed to prove Proposition \ref{thm_tripod} now.

\begin{prooftripod}
Given $s \in \Sigma$, we have
\begin{equation}\label{eq_My}
Q^s=B^{-1}\, \mathrm{diag}(C_{s})\, B.
\end{equation}
 Indeed, note first that $J^{ab}=A^t\, \mathrm{diag}(\pi)\, B$ by equation \eqref{eq_J}. On the other hand, $P^{s}$ is equal to $A^t\, \mathrm{diag}(\pi)\,\mathrm{diag}(C_{s})\, B$ because
  for any $i,j$ we have
  \[P^{s}_{i,j}=p_{i,j,{s}}= \sum_{x \in \Sigma}\pi_{x}\, C_{x,{s}}\,A_{x,i}\,B_{x,j}.\]
Hence, $(J^{ab})^{-1}\, P^{s}= B^{-1}\,\mathrm{diag}(\pi)^{-1} \, (A^t)^{-1}\,A^t\, \mathrm{diag}(\pi)\,\mathrm{diag}(C_{s})\, B$, and equation \eqref{eq_My} follows.

Now, fix $g\in G$. By Lemma   \ref{lem_propietats} $(iii)$  we have
$Q^{\op{g}{s}}=K_g^{-1}Q^sK_g$ for any $s \in \Sigma$. Applying \eqref{eq_My} to $Q^s$ and $Q^{\op{g}{s}}$ we obtain
\begin{equation}\label{eq_Dy}
\mathrm{diag}(C_{\op{g}{s}})=B\,Q^{\op{g}{s}}\,B^{-1}= (B\,K_g^{-1}B^{-1}) \, \mathrm{diag}(C_{s})\, (B\,K_g\,B^{-1}).
\end{equation}
Thus, the matrix $X_g:=B\,K_g^{-1}B^{-1}$ diagonalizes all matrices $\mathrm{diag}(C_{\op{g}{s}})$ (equivalently all  $\mathrm{diag}(C_{s})$, $s\in \Sigma$) and its columns are common eigenvectors to all these diagonal matrices.
We claim that the common eigenspaces to all $\mathrm{diag}(C_{s})$, ${s}\in \Sigma$, have dimension one (even if there are repeated eigenvalues). Indeed, if columns $i,j$ of $X_g$ belong to the same eigenspace for all ${s}$, then looking at the eigenvalues we would have $C_{i,{s}}=C_{j,{s}}$ for all ${s}$. But this is not possible because $C$ has rank $\kappa$, so all its rows are different.

In particular, the columns of $X_g$ are multiples of the standard basis.
As the rows of $B$ are normalized, so are the rows of $B^{-1}$ and hence the rows of $X_g$. From this we obtain that  $X_g$ is a permutation matrix $K_{\sigma_g}$, for a certain permutation $\sigma_g$ (which may depend on $g$ a priori).

Note that the entry $C_{s,s}$ is at row $\op{g}{s}$ of $\mathrm{diag}(C_{\op{g}{s}})$ and it is at row $\sigma_\op{g^{-1}}{s}$ of $K_{\sigma_g}C_{s}K_{\sigma_g}^{-1}$.  As $\mathrm{diag}(C_{\op{g}{s}})=K_{\sigma_g}\mathrm{diag}(C_{s})K_{\sigma_g}^{-1}$ and $C$ belongs to $U$, we have $\op{g}{s}= \op{\sigma_g^{-1}}{s}$.
Thus, for any  ${s}\in \Sigma$ we get $\op{(\sigma_g\circ g)}{s}={s}$ so that $\sigma_g=g^{-1}$. Then we have $X_g=K_g^{-1}$ and
 $K_gBK_g^{-1}=B$. As this argument applies to any $g\in G$, we have that $B$ is $G$-equivariant.

From this we also obtain that $\pi$ is $G$-invariant. Indeed, we consider $\rho=f_a(f_c(p))$, which is $G$-invariant by Lemma \ref{marginalisation}; then  $\pi^t= \rho^t B^{-1}$ and as  $B^{-1}$ is $G$-equivariant, the claim follows.

Finally we have that $A$ and $C$ are also G-equivariant. Indeed, if $J^{bc}$ is the matrix obtained from $f_a(p)$, then it is a $G$-equivariant matrix. Moreover,
$J^{bc}=B^t\,\mathrm{diag}(\pi)\, C$ and $C=\mathrm{diag}(\pi)^{-1}\, B^{-t}\,J^{bc}$ is the product of three $G$-equivariant matrices.
Analogously, exchanging the roles of $a$ and $b$ we can also prove that $A$ is $G$-equivariant.
\end{prooftripod}

\begin{rmk}\rm Note that the previous proposition is still true when we restrict to real parameters and change $U$ to the set of DLC matrices.
\end{rmk}

\subsection{The general case}

Let $T$ be a rooted tree with $n$ leaves and consider a Markov process on it. Given $u,v\in V(T)$, denote by $path(u,v)=(e_1,\ldots,e_m)$ the sequence of edges of $T$ from $u$ to $v$ (so that $u$ is the first node of $e_1$ and $v$ is the last node of $e_m$).

\begin{thm}\label{prop_main}
Let $T$ be a phylogenetic tree and let $p=\phi_T(\pi,\{M^e\}_{e\in E(T)})$ be a point in the image of $\phi_T$. If $p$ is $G$-invariant, the parameters are non-singular, and the transition matrices are in $U$, then $\pi$ is $G$-invariant and all matrices $M^e$ are $G$-equivariant.
\end{thm}

\begin{figure}
	\begin{center}
	 \includegraphics[scale=0.5]{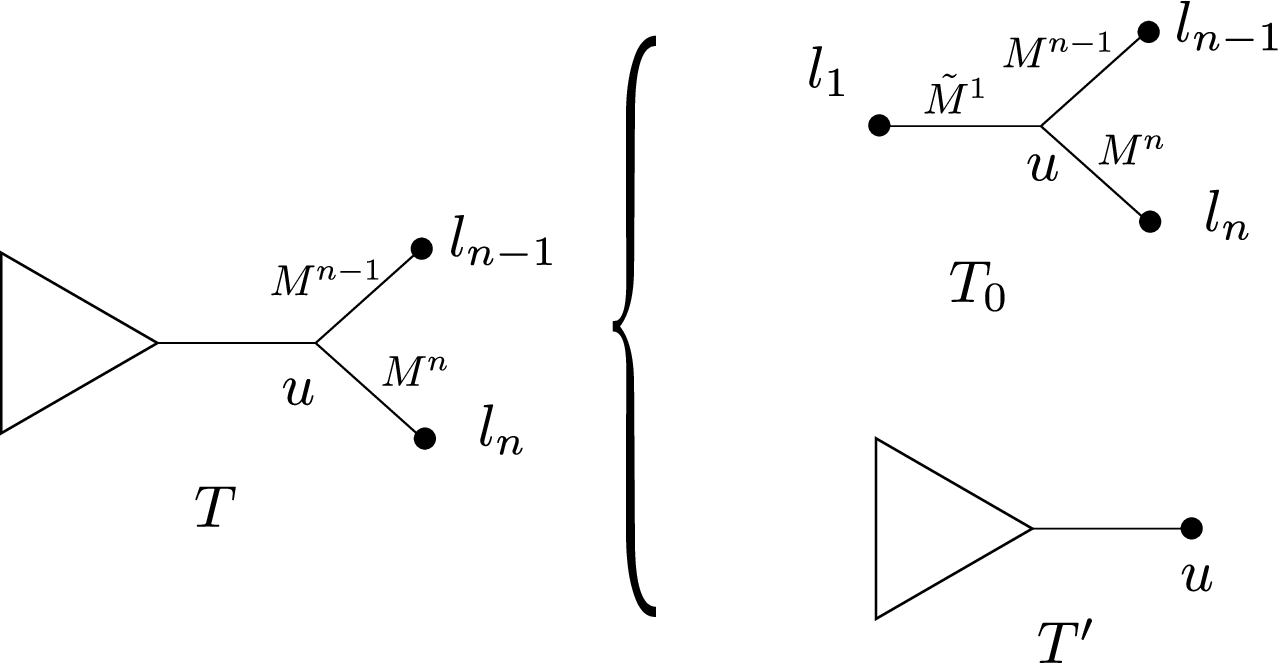}
	\end{center}
\caption{\label{aux_figure} The trivalent tree $T$ has a cherry with leaves $l_{n-1}$ and $l_n$, and parent node $u$. If $l_1$ is another leaf of $T$, and we marginilize over all leaves except for $l_1,l_{n-1},l_n$ we obtain the tripod tree $T_0$ with transition matrices  $\tilde{M}^1,M^{n-1},M^n$. The tree $T'$ is the trivalent tree obtained from $T$ by pruning the leaves $l_{n-1}$ and $l_n$.}

\end{figure}

\begin{proof}
 We proceed by induction on the number  $n$ of leaves  of $T$.

 {For $n=2$ we have {necessarily two nodes $a,b$ and a single edge $e$ (since the tree is not allowed to have degree-2 vertices) and we can root the tree at $a$.} Marginalizing over the leaf $b$ and applying Lemma \ref{marginalisation}, we obtain that $\pi = f_b(p)$ is $G$-invariant, given that $p$ is $G$-invariant. Rewriting $p$ as a matrix $P$ with rows (resp. columns)  labeled by states at leaf $a$ (resp. $b$) we obtain a $G$-equivariant matrix. On the other hand we have $P=\mathrm{diag}(\pi) M^e$. As $\mathrm{diag}(\pi)$ is invertible, we obtain that $M^e=\mathrm{diag}(\pi)^{-1}P$ is the product of two $G$-equivariant matrices and hence it is $G$-equivariant.}

 The case $n=3$ is solved by Proposition \ref{thm_tripod}.

 Now consider $n>3$. 
 We can assume that the tree is trivalent. 
 Indeed, if $T$ is not trivalent, there exists a (not necessarily unique) trivalent tree $T'$ such that $T$ can be obtained by collapsing edges on $T'$ (in other words, $T'$ is a refinement of $T$, see \cite[\S 2.4]{SteelPhylogeny}). In this case, it is enough to consider $T'$ instead of $T$ associating the identity matrix to the edges in $T'$ that are not in $T$. Then $p$ is the image by $\phi_{T'}$ of the new parameters, which still satisfy the hypotheses of the theorem (because the identity matrix is non-singular, $G$-equivariant and lies in $U$).

 Any trivalent tree $T$ has a cherry
 and, by reordering the leaves if necessary, we can assume that this cherry is composed of leaves $l_{n-1}$ and $l_n$.
By Remark \ref{change_root_equivariance}, we can also assume that the tree is rooted at the parent node $u$ of $l_{n-1}$ and $l_n$ (see Figure \ref{aux_figure}).
%
 %
 %
 By marginalizing $p$ over all leaves except for $l_1$, $l_{n-1}$ and $l_n$, we obtain the tensor $p_0=\phi_{T_0}(\pi,{\tilde{M^1}},{M^{n-1}},{M^n})$ where  $T_0$ is the tripod tree with leaves $l_1$, $l_{n-1}$ and $l_n$ (see Figure \ref{aux_figure}), and ${\tilde{M^1}}=\prod_{e\in path(u,l_1)} {M^e}$ is the transition matrix corresponding to the concatenation of the edges in $path(u,l_1)$. Note that $p_0$ is $G$-invariant by Lemma \ref{marginalisation}. Since ${M^{n-1}}$ lies in $U$, we can apply Proposition \ref{thm_tripod} to conclude that $\pi$ is $G$-invariant and both ${M^{n-1}}$ and ${M^n}$ are $G$-equivariant. \\
Next, consider the tree $T'$ obtained by pruning the leaves $l_{n-1}$ and $l_n$ {(that is, remove from $T$ the two pendant edges adjacent to leaves $l_n$ and $l_{n-1}$)}. Write $L'=L(T)\setminus \{l_{n-1},l_n\}$ so that $L(T')=L'\cup \{u\}$.  Note that as $T$ was assumed to have no nodes of degree two (by definition of phylogenetic tree), $T'$ also satisfies this assumption. To finish the proof it is enough to prove that the tensor $p'=\phi_{T'}(\pi,\{{M^e}\}_{e\in E(T')})$ is $G$-invariant and apply the induction hypothesis to deduce that all transition matrices $\{{M^e}\}_{e\in E(T')}$ are $G$-equivariant.
Given states $s_1,\ldots,s_{n-2}\in \Sigma$ associated to the leaves in $L'$, we denote $\mathbf{s}=(s_1,\ldots,s_{n-2})$ and $\varphi_{\mathbf{s}}(y)=p'_{s_1,\ldots,s_{n-2},y}$. If $g\in G$, we write $g\cdot \mathbf{s}=(g s_1,\ldots,g  s_{n-2})$.
Keeping the notation of Proposition \ref{thm_tripod}, write $P^{\mathbf{s}}$ for the $\kappa \times \kappa$-matrix whose entries are defined by
 \begin{eqnarray*}
 P^{\mathbf{s}}(j,k) = p_{s_1,\ldots,s_{n-2},j,k}, \quad j,k\in \Sigma.
 \end{eqnarray*}

From the parameterization we have
\begin{eqnarray*}
	P^{\mathbf{s}}(j,k) = p_{\mathbf{s},j,k} = \sum_{i\in \Sigma} p'_{\mathbf{
			s}',i} \cdot {M^{n-1}_{i,j}} \cdot  {M^n_{i,k}} = \sum_{i\in \Sigma} \varphi_{\mathbf{s}}(i) \cdot {M^{n-1}_{i,j}} \cdot  {M^n_{i,k}}
\end{eqnarray*}
In matrix notation, if $\Phi_{\mathbf{s}}=\mathrm{diag} ((\varphi_{\mathbf{s}}(i))_{i\in \Sigma})$, then
$P^{\mathbf{s}}= {(M^{n-1})^t} \, \Phi_{\mathbf{s}} \, {M^n}$.
Since the parameters are nonsingular, we get that
\begin{eqnarray*}
	\Phi_{\mathbf{s}}= {(M^{n-1})^{-t}} P^{\mathbf{s}} {{M^n}^{-1}}
\end{eqnarray*}
Moreover, since ${M^{n-1}}$ and ${M^n}$ are $G$-equivariant, we have that
\begin{eqnarray}\label{eq_Ks}
	K_g \Phi_{\mathbf{s}} K_g^{-1}=
K_g \; {(M^{n-1})^{-t}} P^{\mathbf{s}} {(M^n)^{-1}} \; K_g^{-1}=
    {(M^{n-1})^{-t}} K_g \; P^{\mathbf{s}}K_g^{-1} \; {(M^n)^{-1}}.
\end{eqnarray}

 Note that for every $g\in G$, we have $K_g P^{\mathbf{s}} K_g^{-1}=P^{g^{-1}\cdot \mathbf{s}}$ (indeed, the $G$-invariance of $p$ gives that the $(j,k)$-entry of $K_g P^{\mathbf{s}} K_g^{-1}$ is
 	$ P^{\mathbf{s}}(gj,gk) = p_{\mathbf{s},gj,gk} = p_{g^{-1}\cdot \mathbf{s},j,k}$).
 Thus, in \eqref{eq_Ks} we have
 $$K_g \Phi_{\mathbf{s}} K_g^{-1}={(M^{n-1})^{-t}} P^{g^{-1}\mathbf{s}}  {(M^n)^{-1}} = \Phi_{g^{-1} \mathbf{s}}.$$

Finally, for each $i\in \Sigma$, we have that
\begin{eqnarray*}
	p'_{g \cdot \mathbf{s},g \cdot i} = \varphi_{g \cdot \mathbf{s}}(g \cdot i) = \Phi_{g  \cdot \mathbf{s}} \, (g \cdot i,g \cdot i) =
	\big ( K_g \Phi_{g \cdot \mathbf{s}} K_g^{-1} \big )\, (i,i) =\Phi_{\mathbf{s}} (i,i) = \varphi_{\mathbf{s}}(i) = p'_{\mathbf{s},i},
\end{eqnarray*}
showing that the tensor $p'$ is $G$-invariant.
\end{proof}

\begin{rmk}\rm Up to here, the results in this section still hold if we restrict to parameters in $\mathbb{R}$ instead of $\mathbb{C}$, which makes more sense biologically speaking. We could even assume that the transition matrices are non-negative to have probability distributions. Henceforth, we consider algebraic varieties and we need to work with the complex field.
\end{rmk}

Before stating the main result, we recall a definition from \cite{casfermich2015}. A point $p \in V_T$ is said to be a point of \emph{no evolution} if $p=\phi_T(\pi,\mathbf{Id})$ for some $\pi \in W$ ($\mathbf{Id}$ means a collection of identity matrices, one for each edge). We write $W_{\neq 0}$ for the subset of vectors in $\mathcal{W}$ with non-zero coordinates.

\begin{thm}\label{thm_main}
	For any trivalent phylogenetic tree $T$ and any subgroup $G \leq \mathfrak{S}_{\kappa}$, the variety $V_T^G$ equals the irreducible component of the intersection $V_T \cap \LL^G$ that contains $\Im \phi_T^G$.
 Moreover, $V_T^G$ is the unique irreducible component of $V_T\cap \LL^G$ that contains points of no evolution $p^0_G=\phi_T(\pi,\mathbf{Id})$ for $G$-invariant vectors $\pi\in W_{\neq 0}$.
\end{thm}

\begin{proof}
 The set $\mathcal{U}=W_{\neq0} \times \prod_{e \in E(T)} U$  is Zariski-dense in $Par(T)$. Thus $\phi_T(\mathcal{U})$ is Zariski-dense in $\phi_T(Par(T))$, which is Zariski-dense in $V_T$. By Theorem \ref{prop_main}, we know that
	\begin{equation}\label{eq_phi}
	\phi_T(\mathcal{U})\cap \LL^G= \phi_T^G(\mathcal{U}\cap Par_G(T)).
	\end{equation}


 Consider a $G$-invariant distribution $\pi$ with $\pi_{x}\neq 0$ for any $x\in \Sigma$, and define $p^0_G=\phi_T^G(\pi;\mathbf{Id})\in \phi_T(\mathcal{U})$ the corresponding point of no evolution.  From \cite[Theorem 5.4 and Appendix 6.2]{CFM} we deduce that $p^0_G$ is a smooth point in $V_T$ (here is where we have to use the trivalent hypothesis, as smoothness has not yet been proven for no evolution points on non-trivalent trees).
Thus, there is a local biholomorphism between an open neighbourhood $\mathcal{W}$ (in standard complex topology) of $(\pi,\mathbf{Id})$ in $\mathcal{U}$ and an open neighbourhood $\phi_T(\mathcal{U}) \cap \mathcal{O}$ of $p^0_G$ in $\phi_T(\mathcal{U}) \subset V_T$, where $\mathcal O$ is a convenient open set of $\mathcal{L}$ (in standard topology) containing $p^0_G$: $\mathcal{W} \cong \phi_T(\mathcal{W})=\phi_T(\mathcal{U})\cap \mathcal{O}$.
Intersecting with $\mathcal{L}^G$ we have that $p^0_G$ belongs to
\begin{eqnarray*}
V_T \cap \mathcal{O} \cap \mathcal \LL^G
= \phi_T(\mathcal{W}) \cap \LL^G = \phi_T^G(\mathcal{W} \cap Par_G(T) )
\end{eqnarray*}
where the last equality is obtained from (\ref{eq_phi}) restricted to $\mathcal{W}$.


If $C_1,\dots,C_m$ are the irreducible components of $V_T\cap \mathcal{L}^G$ that meet $\mathcal{O}$, then taking the Zariski closure in the above expression we have
\[C_1\cup \dots \cup C_m= \overline{\phi_T^G(\mathcal{W}\cap Par_G(T))}=V_T^G.\] As $V_T^G$ is irreducible, we get $m=1$ and $V_T^G$ coincides with an irreducible component of $V_T\cap \mathcal{L}^G$. Moreover, since $\mathcal{W} \cap Par_G(T)$ is irreducible and biholomorphic to its image, $p^0_G$ is contained in a unique irreducible component of $V_T^G\cap \LL^G$.
\end{proof}

As a consequence of Theorem \ref{prop_main} we also obtain that the restriction to a certain equivariant submodel can be done step by step by considering intermediate submodels:

\begin{corollary}\label{cor_subgrup} Consider two groups $G_1$, $G_2$ such that $G_1\leq G_2\leq \mathfrak{S}_{\kappa}$ and let $U$ be the open subset of the space of matrices introduced above. Then,
if $\mathcal{U}=W_{\neq 0} \times \prod_{e \in E(T)} U$ and $T$ is a trivalent phylogenetic tree, we have
\begin{enumerate}
    \item[(i)] $\phi_T^{G_2}(\mathcal{U}\cap Par_{G_2}(T))=\phi_T^{G_1}(\mathcal{U}\cap Par_{G_1}(T))\cap \LL^{G_2}$.
\item[(ii)]  $V_T^{G_2}$ is the irreducible component of $V_T^{G_1} \cap \LL^{G_2}$ that contains $\Im \phi_T^{G_2}$.
\end{enumerate}
\end{corollary}

\begin{proof}
   \begin{enumerate}
   \item[(i)] Note that the definition of the open set $\mathcal{U}$ is independent of the subgroup $G$. From (\ref{eq_phi}), we have that $\phi_T^{G_1}(\mathcal{U}\cap Par_{G_1}(T)) = \phi_T(\mathcal{U})\cap \LL^{G_1}$.
    Since $\LL^{G_2}\subseteq \LL^{G_1}$, we deduce that
    $\phi_T^{G_1}(\mathcal{U}\cap Par_{G_1}(T)) \cap \LL^{G_2} = \phi_T(\mathcal{U})\cap \LL^{G_2} = \phi_T^{G_2}(\mathcal{U}\cap Par_{G_2}(T))$, the last equality again by (\ref{eq_phi}).
    \item[(ii)] 
    By Theorem \ref{thm_main}, $V_T^{G_1}$ is an irreducible component of $V_T\cap \mathcal{L}^{G_1}$, so the irreducible components of $V_T^{G_1}\cap  \mathcal{L}^{G_2}$ are a subset of the irreducible components of $V_T\cap \mathcal{L}^{G_2}$ (note that $\mathcal{L}^{G_2}\subseteq \mathcal{L}^{G_1}$).
    Since $\Im \phi_T^{G_2} \subseteq \Im \phi_T^{G_1}$, the irreducible component of $V_T\cap \LL^{G_2}$ that contains $\Im \phi_T^{G_2}$ equals the irreducible component of $V_T^{G_1} \cap \LL^{G_2}$ that contains $\Im \phi_T^{G_2}$. Finally,
    Theorem \ref{thm_main} claims that this component is $V_T^{G_2}$.
      \end{enumerate}
\end{proof}

\begin{rmk}\label{rmk_edges}\rm
%
For phylogenetic reconstruction purposes it is important to note that, if $T'$ is another tree with leaf set $L$ not obtained by collapsing an edge in $T$, then $V^G_{T'}$ is not an irreducible component of $V_T\cap \LL^G$. Indeed, there is an edge split $A|B$ on $T$ that is not an edge split in $T'$. {Thus, if $p \in \Im \phi_{T'}^G$ is a generic point, by Theorem \ref{thm_AR} we get $p\notin V_T$.}
So, $V_{T'}^G$ is not contained in $V_T$ and cannot be contained in $V_T\cap \LL^G$ either.

Theorem \ref{thm_main}
implies that if we have a  data point, say $\hat{p}$, which is an approximation of a theoretical distribution $p\in V_T^G$ generated under a $G$-equivariant model (with generic parameters), it is enough to verify that $\hat{p}$ lies on (or is close to) the variety $V_T$ to deduce that $T$ is the closest tree topology for the data. That is, we do not need specific generators of $I_T^G$.
%

Moreover, if $\hat{p}$ is a data point close to a generic point of no evolution $p^0_G\in V_T^G$ for a trivalent tree $T$ (as in the case of biological data), Theorem \ref{thm_main} ensures that $\hat{p}$ is close to a unique irreducible component of $V_T\cap\mathcal{L}^G$, namely $V_T^G$. The standard simplex where biological data lies can meet  $V_T\cap\mathcal{L}^G$ in other irreducible components, but none of them is close to points of no evolution.
\end{rmk}

\section{Edge invariants for equivariant models}

{The aim of this section is to prove that, for a tensor $p\in \LL^G$, imposing rank conditions on the decomposition of a flattening matrix into isotypic components (in the sense of Draisma-Kuttler work \cite{Draisma}) is equivalent to imposing the vanishing of $(\kappa+1)\times (\kappa+1)$ minors if $p$ belongs to a certain open subset. The proof follows the ideas we used in the examples of Section 3 and the result allows obtaining equations for $V_T^G$ by saturating an ideal generated by $(\kappa+1)\times (\kappa+1)$ minors with respect to a certain ideal (see Corollary \ref{cor:saturation}). 

To introduce Draisma and Kuttler's framework and to prove our result for any $G$-equivariant model, first we need to recall some tools of representation theory of groups.}

Given a permutation group $G\leq \mathfrak{S}_{\kappa}$, denote by
	$N_1,\ldots, N_s$ the inequivalent irreducible representations of $G$ and by $d_i=\dim \, N_i$ ($i=1,\ldots,s$) their dimensions.
	Given a linear representation $\rho:G \rightarrow GL(V)$, Maschke's theorem {(see for instance \cite{james2001representations})} establishes a decomposition
	\begin{eqnarray*}
		V=\oplus_{i=1}^s V_i,
	\end{eqnarray*}
	where the $V_i$ are the \emph{isotypic components}. Each $V_i$ is a $G$-submodule of $V$ isomorphic to several copies of the irreducible representation $N_i$: $V_i\cong N_i\otimes \CC^{m_i(V)}$. The value $m_i(V)$ is the \emph{multiplicity} of $V$ relative to the irreducible representation $N_i$.
Schur's lemma establishes that
\begin{eqnarray*}
\Hom_G(N_i,N_j)= \left \{ \begin{array}{cc}
\CC\cdot  \mathrm{Id} & \mbox{ if }i=j \\
0 & \mbox{otherwise.}
\end{array}
\right .
\end{eqnarray*}
We adopt the notation of \cite{CFM} and write $\F_i(V)\subseteq V$ for the subspace of $V$ given by the image of a particular {nonzero} element $v_i\in N_i$ by all $G$-equivariant homomorphisms from $N_i$ to $V$:
	\begin{eqnarray*}
		\F_i(V)=\{f(v_i)\mid f\in \Hom_G(N_i,V)\} \cong \CC^{m_i(V)}.
	\end{eqnarray*}
These subspaces represent the whole isotypic component $V_i$ as  $V_i\cong N_i\otimes \F_i(V)$.
If $V'$ is another $G$-module (or even the same $V$), any $G$-equivariant map 
$h:V \rightarrow V'$ induces by restriction to $\F_i(V)$ a linear map $h_i:\F_i(V) \rightarrow \F_i(V')$. We obtain a natural map
\begin{eqnarray} \label{iso_G}
 \Hom_G(V,V')\rightarrow \bigoplus_i \Hom_{\CC}(\F_i(V),\F_i(V')),
\end{eqnarray}
which is actually a linear isomorphism
%
(see Remark 4.1 of \cite{CFM}). Indeed, since $V_k\cong N_k\otimes \F_k(V)$ and $V'_k\cong N_k\otimes \F_k(V')$, we have
\begin{eqnarray*}
 \Hom_G(V_k,V_k') \cong 
 \Hom_{\CC} (\F_k(V),\F_k(V'))\otimes \Hom_G(N_k,N_k),
\end{eqnarray*}
which can be identified with $\Hom_{\CC}(\F_k(V),\F_k(V'))$ since $\Hom_G(N_k,N_k) = \CC\cdot  \mathrm{Id}_{N_k}$ (by Schur's lemma). This allows us to identify every $G$-equivariant map $h:V\rightarrow V'$ with a collection of linear maps $(h_1,\ldots,h_s)$,  $h_k:\F_k(V)\rightarrow \F_k(V')$, according to the decomposition
\begin{eqnarray*}
h=\sum_{k=1}^s h_k \otimes \mathrm{Id}_{N_k}.
\end{eqnarray*}

\vspace{3mm}
Back to the case of our primary interest, from now on we only consider linear representations $V=\otimes^r W$, $r\in \mathbb{N}$, induced by the permutation representation of $G$. 
In the simplest case, when $V=W$ is the restriction of permutation representation to the elements of $G$,
we denote $\F_k(W)$ as $\F_k$ and $m_k(W)$ as $m_k$. We assume that the irreducible representations $N_1,\ldots,N_s$ of $G$  are ordered so that
$m_k>0$ if $k=1,\dots,l$, and $m_k=0$ if $k\geq l+1$, and we denote as $\mathbf{m}=(m_1,\ldots,m_s)$ the collection of multiplicities.
%

If $A|B$ is a bipartition of $L$,
we denote
\begin{eqnarray*}
	W_A= \otimes_{u\in A} W_u \qquad W_B= \otimes_{v\in B} W_v
\end{eqnarray*}
and write $\F_k^A$ (resp. $\F_k^B$) for the subspaces $\F_k(W_A)$ (resp. $\F_k(W_B)$). {By choosing bases of $W_A$ and $W_B$ adapted to a decomposition into isotypic components, $flat_{A|B}$ becomes \emph{block-diagonal.}} 
%
%
Then the above isomorphism (\ref{iso_G}) can be written as
\begin{eqnarray}\label{aux:Schur}
		Tf_{A|B}: \LL^G=(W_A\otimes W_B)^G \; \cong \;\bigoplus_{k=1}^s \Hom_{\CC}(\F_k^A,\F_k^B) \,.
\end{eqnarray}
{We will refer to this isomorphism as the \emph{thin flattening relative to the bipartion $A|B$}. It  maps every tensor in $\LL^G$ to a collection of linear maps $\{h_k:\F_k^A \rightarrow \F_k^B\}_{k=1,\dots,s}$. By fixing bases of the linear spaces $\F_i^A$ and $\F_j^B$, $Tf_{A|B}(p)$ can be represented as a block-diagonal matrix. For simplicity and by a slight abuse of terminology, we will refer to this matrix using the same notation $Tf_{A|B}(p)$ (see \cite{CF11} for further details and the forthcoming Example 5.3). 

Draisma and Kuttler proved in \cite{Draisma} that, if $p\in \Im(\phi_T)$ and $A|B$ is an edge split of $T$, then 
\begin{equation}\label{eq:rkTf}
\rk Tf_{A|B}(p)\leq \mathbf{m},    
\end{equation}
where $\rk$ indicates the collection of ranks of the diagonal subblocks.
}

{
\begin{ex}\rm
\label{ex_K80}
    Here we illustrate the definition of thin flattening with the cases studied in subsections \ref{sec:k81} and \ref{sec:k80} and using the notation of \cite[Example 5.3]{CF11} and \cite{CFM}.  We consider $n=4$, $L=\{1,2,3,4\}$, the bipartition $A|B$ given by $A=\{1,2\}$, $B=\{3,4\}$, and call $\bp_{x_1x_2x_3x_4}$ the Fourier coordinates of  $p \in \LL^G$ as in Section 3. For $G=K81$, there are precisely four irreducible representations of $G$ and all have dimension 1: we denote them by $N_{\omega_{\a}}, N_{\omega_{\c}}, N_{\omega_{\g}}$ and $N_{\omega_{\t}}$. %
      %
The explicit description of the isotypic components of $W\otimes W$ is given by $(W\otimes W)_{\a}=\langle \bar{\a}\otimes \bar{\a},\bar{\c}\otimes\bar{\c},\bar{\g}\otimes\bar{\g},\bar{\t}\otimes\bar{\t} \rangle$, $(W\otimes W)_{\c}=\langle \bar{\a}\otimes \bar{\c},\bar{\c}\otimes\bar{\a},\bar{\g}\otimes\bar{\t},\bar{\t}\otimes\bar{\g} \rangle$, $(W\otimes W)_{\g}=\langle \bar{\a}\otimes \bar{\g},\bar{\c}\otimes\bar{\t},\bar{\g}\otimes\bar{\a},\bar{\t}\otimes\bar{\c} \rangle$, $(W\otimes W)_{\t}=\langle \bar{\a}\otimes \bar{\t},\bar{\c}\otimes\bar{\g},\bar{\g}\otimes\bar{\c},\bar{\t}\otimes\bar{\a} \rangle$ (where subindices $x$ refer to representation $N_{\omega_{x}}$ ). The block-diagonal matrix $flat_{12|34}(\bar{p})$ of \eqref{flat_K3} is obtained by changing $flat_{12|34}(p)$ to this basis of $W\otimes W$. In this case we have $\mathcal{F}_{k}(W\otimes W)\cong \CC^4$, $k=\omega_{\a}$, $\omega_{\c}$, $\omega_{\g}$, $\omega_{\t}$, and since the irreducible representations have dimension one, the thin flattening matrix coincides with $flat_{12|34}(\bar{p})$.

    Now, we change the group and consider $G=K80$ as in subsection \ref{sec_quartetsK80}. 
    There are 5 irreducible representations of the group $K80$, four of dimension one $N_{\omega_1}$, $N_{\omega_2}$, $N_{\omega_3}$, $N_{\omega_4}$, and one of dimension two $N_{\omega}$ (see \cite[Ex 5.4]{CF11} for their description in terms of the character table of the group and for the justification on the claims in this example). 
    By using representation theory, we obtain the explicit decomposition of $W\otimes W$ into the isotypic components according to the new group: 
    \begin{eqnarray*}
        (W\otimes W)_{\omega_1} & = & \langle \bar{\a} \otimes \bar{\a},\bar{\g} \otimes \bar{\g},\bar{\c} \otimes \bar{\c}+\bar{\t} \otimes \bar{\t} \rangle, \quad \\
        (W\otimes W)_{\omega_2}& = & \langle \bar{\c} \otimes \bar{\t}-\bar{\t} \otimes \bar{\c}\rangle, \quad \\
        (W\otimes W)_{\omega_3}& = & \langle \bar{\a} \otimes \bar{\g},\bar{\g} \otimes \bar{\a},\bar{\c} \otimes \bar{\t}+\bar{\t} \otimes \bar{\c} \rangle, \quad  \\
        (W\otimes W)_{\omega_4}& = & \langle \bar{\c} \otimes \bar{\c}-\bar{\t} \otimes \bar{\t}\rangle, \quad \\
        (W\otimes W)_{\omega}& = & \langle \bar{\a} \otimes \bar{\c}, \bar{\c} \otimes \bar{\a}, \bar{\g} \otimes \bar{\t}, \bar{\t} \otimes \bar{\g}, \bar{\a} \otimes \bar{\t}, \bar{\t} \otimes \bar{\a}, \bar{\g} \otimes \bar{\c}, \bar{\c} \otimes \bar{\g} \rangle, \, \,  
    \end{eqnarray*}
 and   
 \begin{gather*}
     \F_{\omega_1} (W\otimes W)\cong \CC^3, \F_{\omega_2} (W\otimes W)\cong \CC , \F_{\omega_3} (W\otimes W)\cong \CC^3,  \F_{\omega_4} (W\otimes W)\cong \CC,\\
     \F_{\omega} (W\otimes W)\cong \CC^4.
 \end{gather*}

   By taking this basis of $W\otimes W$,
$flat_{A|B}(p)$ becomes a 5-block diagonal matrix with blocks $S_1\in M_{{3\times 3}}{(\CC)}$, $S_2\in M_{{1\times 1}}{(\CC)}$, $S_3\in M_{{3\times 3}}{(\CC)}$, $S_4\in M_{{1\times 1}}{(\CC)}$ and $S\in M_{{8\times 8}}{(\CC)}$. By comparing with the isotypic components obtained for $K81$, we realize that
    \begin{eqnarray*}
        (W\otimes W)_{\a} & =& (W\otimes W)_{\omega_1} \oplus (W\otimes W)_{\omega_4}, \\
        (W\otimes W)_{\g} & =& (W\otimes W)_{\omega_2} \oplus (W\otimes W)_{\omega_3}.
    \end{eqnarray*}
    The first decomposition implies that the block $B_{\a}$ of (\ref{flat_K3}) becomes in turn a 2-block diagonal matrix, with blocks $S_1$ and $S_4$. Similarly, the second decomposition implies that $B_{\g}$ also becomes a 2-block diagonal matrix, with blocks $S_2$ and $S_3$. This is analogous to imposing the linear equations of the first three rows of Remark \ref{modeq:K80}. 
    %
%

    On the other hand, the 16 linear remaining equations of Remark \ref{modeq:K80} (or equivalently, the redundancy between the blocks $B_{\c}$ and $B_{\t}$) are translated into the fact that the $8\times 8$ matrix $S$ corresponds to  $\Hom_G((W\otimes W)_{\omega},(W\otimes W)_{\omega})$, which the thin flattening \eqref{aux:Schur} maps to $\Hom_{\CC}(\F_{\omega}^A,\F_{\omega}^B)$. This is, by restricting to tensors in $K80$, $S$ can be reduced to a  $4\times 4$ matrix $S'$ (analogously, to one of the two copies $B_{\c}$ or $B_{\t}$). The thin flattening in this case is the block-diagonal matrix \[Tf_{A|B}(p)=\ddd{S_1,S_2,S_3,S_4,S'}\]
    and \eqref{eq:rkTf} gives  $\rk Tf_{A|B}(p^T)\leq \mathbf{m}=(1,0,1,0,1)$ if $A|B$ is an edge split of $T$ ($W=N_{\omega_1}\oplus N_{\omega_1} \oplus N_{\omega}$ in this case). For instance, the conditions $\rank(S_2)\leq 0$ and $\rank(S_2)\leq 0$ are equivalent to the two linear topology invariants found in Subsection \ref{sec_quartetsK80}.
 \end{ex}
}


\begin{rmk}\label{rmk_ones}
 Note that if $k\leq l$, both $\F_k^A$ and $\F_k^B$ are non-zero as they contain $H_k^A:= \textbf{1}\otimes  \ldots\otimes\textbf{1}\otimes \F_k$ and $H_k^B:= \textbf{1}\otimes \ldots\otimes\textbf{1}\otimes \F_k$, respectively, where $\textbf{1}=\sum_{i\in \Sigma} i$. In particular, $\dim \F_k^A$ and $\dim \F_k^B$ are strictly positive for $k\leq l$.

 These subspaces play a special role due to the following reason. In terms of coordinate rings, the map $f_l$ introduced in Section 4 can be easily described (as explained in \cite[\S 4]{CFM}).  Indeed, if $l=l_1$ (to simplify notation) the dual of the marginalization map $f_l$  is
 \begin{equation}\label{eq_dualmarg}
     \begin{array}{rcl}
          f_ l^*: \otimes_{u\neq l} W_u&
          \longrightarrow &\otimes_{u\in L} W_u\\
     t & \mapsto & \textbf{1} \otimes t
     \end{array}.
 \end{equation}
 Note that this map is basis independent and restricts to $G$-invariant tensors.
\end{rmk}

The main result of this section is the following. 

\begin{thm}\label{thm_edges}
	Let $A|B=\{i_1,\ldots,i_a\} \mid \{j_1,\ldots,j_b\}$ be a bipartition of $L$. Let $p\in \LL^G_n$ be a tensor
	such that the marginalization
	\begin{eqnarray*}
 p'=(f_{i_1}\circ \cdots \circ f_{i_{a-1}} \circ f_{j_1} \circ \ldots \circ f_{j_{b-1}})(p)\in (W_{i_a}\otimes W_{j_b})^G
	\end{eqnarray*}
	has maximal rank as a homomorphism in
 $\Hom_G(W_{i_a},W_{j_b})\cong \bigoplus_k \Hom_{\CC} (\F_{k},\F_{k})$
	(that is, it has rank $\mathbf{m}$).
	Then, \[\rank flat_{A|B}(p)\leq \kappa \textrm{ if and only if } \rk Tf_{A|B}(p)\leq \mathbf{m}.\]
\end{thm}

\begin{proof} It is immediate to prove that if $\rk Tf_{A|B}(p)\leq \mathbf{m}$, then $\rank flat_{A|B}(p)\leq \kappa$. We proceed to prove the converse.

We have an isomorphism
\begin{eqnarray*}
flat_{A|B}: \LL \rightarrow \Hom_{\CC}(W_A,W_B)
\end{eqnarray*}
that maps $p$ to its flattening $flat_{A|B}(p)$. On the other hand,
$p$ belongs to $\LL^G$, which by (\ref{aux:Schur}) is isomorphic to $\bigoplus_{k=1}^s \Hom_{\CC}(\F_k^A,\F_k^B)$
via the map $Tf_{A|B}$.
%
The connection between both maps is well described by
the following commutative diagram:
\[\xymatrix{ \LL \quad \ar[r]^-{flat_{A|B}} & \quad \Hom_{\CC}(W_A,W_B) 
\; = &
 \hspace{-10mm}\bigoplus_{x,y}\Hom_{\CC}(\F_x^A,\F_y^B)\otimes \Hom_{\CC}(N_x,N_y) \\
  \LL^G  \quad
 \ar@{^{(}->}[u]
 \ar[rr]^-{Tf_{A|B}} & \quad & \bigoplus_{k=1}^s
 \Hom_{\CC}(\F_k^A,\F_k^B)
 \ar@{^{(}->}[u]^{\theta}
}\]

where horizontal arrows correspond to flattening and thin flattening, respectively. Vertical arrows correspond to the natural inclusion (left) and the natural injection $\theta$ that can be described as follows.
%
Fix $k=1,\ldots,s$, and let $h_k \in \Hom_{\CC} (\F_k^A, \F_k^B)$, which naturally corresponds to $h_k \otimes \mathrm{Id} \in \Hom_{\CC} (\F_k^A, \F_k^B) \otimes \Hom_G(N_k,N_k)$.
This space is isomorphic to $\Hom_{\CC} (\F_k^A\otimes N_k,\F_k^B\otimes N_k),$ which is naturally immersed as a subspace in the arrival space of $\theta$,
$$\bigoplus_{x,y}\Hom_{\CC}%
(\F_x^A,\F_y^B))\otimes \Hom_{\CC}(N_x,N_y).
$$
Note that the rank of $h_k \otimes \mathrm{Id} \in \Hom_{\CC} (\F_k^A\otimes N_k,\F_k^B\otimes N_k)$ is equal to $d_k \, \rank h_k$.

Now if $Tf_{A|B}(p) = (h_1,\dots,h_s)$, according to this commutative diagram we have
\[\rank flat_{A|B}(p) = d_1 \,\rank h_1 + \dots + d_s \,\rank h_s,\]which, by assumption, is smaller than
$\rank flat(p) \leq \kappa=m_1\,d_1+\cdots+m_l\,d_l$. We conclude that
\begin{equation}\label{ineq}
 \sum_{k=1}^l d_k \rank h_k  \leq d_1\,\rank h_1+\cdots +d_s\, \rank h_s\leq m_1\,d_1+\cdots+m_l\,d_l.
\end{equation}
On the other hand, the hypothesis of Theorem \ref{thm_edges} gives that $\rank h_k \geq m_k$ for all $1\leq k\leq l$. 
Indeed, in the notation of Remark \ref {rmk_ones}, restricting $h_k$ to $H_k^A\subset \F_k^A$ and projecting to  $H_k^B$ corresponds to the $k$-th component of the tensor $p'\in \oplus_k\Hom_{\CC} (\F_{k},\F_{k})$ in the statement, which has rank $m_k$ by hypothesis.
Inequalities \eqref{ineq}  force $\rank h_k=m_i$ for every $1\leq k\leq l$ and $\rank h_k=0$ for $k>l$. 
\begin{eqnarray*}
 \rk Tf_{A|B}(p) = (\rank h_1,\ldots,\rank h_s) \leq (m_1,\ldots,m_l,0,\dots, 0) =\mathbf{m}.
\end{eqnarray*}
\end{proof}



\begin{ex}\label{sec5_ExK81}\rm
Here we illustrate the hypotheses of the above theorem with the case $G=K81$ studied in Example \ref{ex_K80}.
For the permutation representation $W$ we have $\mathcal{F}_{k}\cong \CC$ for $k\in Y=\{\omega_\a,\omega_\c,\omega_\g,\omega_\t\}$ and $W= (\mathcal{F}_{\omega_\a}\otimes N_{\omega_\a})\oplus (\mathcal{F}_{\omega_\c}\otimes N_{\omega_\c}) \oplus (\mathcal{F}_{\omega_\g}\otimes N_{\omega_\g}) \oplus (\mathcal{F}_{\omega_\t}\otimes N_{\omega_\t}) $ and hence $\mathbf{m}= (1,1,1,1)$. 


Consider $A=\{1,2\}$, $B=\{3,4\}$ so that the marginalization in the hypotheses of the theorem is over leaves 1 and 3:
$p'_{x,y}=\sum_{i_1,i_3}p_{i_1 x i_3 y}.$ By \eqref{eq_dualmarg}, the dual of this marginalization map sends any $v_2\otimes v_4$ to $\ba\otimes v_2\otimes \ba\otimes v_4$ (because  $\textbf{1}=\ba$) and this description is basis independent. Thus, translated into Fourier coordinates, this marginalization map  is:
\[ \begin{array}{rcl}
     \LL^G & \longrightarrow & \Hom_G(W,W)\cong \bigoplus_{k\in Y}
     \Hom(\mathcal{F}_{k},\mathcal{F}_{k})\\
     \bp & \mapsto& 
     \begin{pmatrix}
         \bp_{\texttt{AAAA}}&0&0&0\\
         0& \bp_{\texttt{ACAC}}&0&0\\
         0&0&\bp_{\texttt{AGAG}}&0\\
         0&0&0&\bp_{\texttt{ATAT}}
     \end{pmatrix}
\end{array}.
\]
The hypothesis of Theorem \ref{thm_edges} requires this block diagonal matrix to have maximal rank, which is equivalent to the condition
$\bar{p}_{\texttt{AAAA}}\neq 0$, $\bar{p}_{\texttt{ACAC}}\neq 0$, $\bar{p}_{\texttt{AGAG}}\neq 0$, $\bar{p}_{\texttt{ATAT}}\neq 0$ that we gave in Section \ref{sec:k81}, see \eqref{eq:O}.
\end{ex}

Below we write some consequences of Theorem \ref{thm_edges}. 
{As introduced in section 4 for Theorem \ref{thm_main}, points of no evolution are the image by $\phi_T$ of a Markov process on a tree whose transition matrices are equal to the identity. In biology,  as mutations are generally considered rare events in nature, transition matrices are expected to be close to the identity matrix and a hidden Markov process on tree with biologically meaningful parameters gives rise to a distribution close to a point of no evolution. 
}

\begin{corollary}
 Let $A|B=\{i_1,\ldots,i_a\} \mid \{j_1,\ldots,j_b\}$ be a bipartition of $L$. 
There exists a non-empty  Zariski open set $\mathcal{O}$ of $\LL^G_n$ such that if $p\in \mathcal{O}$ then
\begin{eqnarray*}
\rank flat_{A|B}(p)\leq \kappa \quad \Leftrightarrow \quad \rk Tf_{A|B}(p)\leq \mathbf{m}.
\end{eqnarray*}
Moreover, $\mathcal{O}$ contains all points of no evolution $p\in \LL^G$ such that ${p_{i\dots i}\neq 0}, i\in \Sigma$.
\end{corollary}

\begin{proof}
By  \cite[Lemma 5.5]{CFM}, all generic points of no evolution $p\in \LL^G$ with ${p_{i\dots i}\neq 0}$, $\forall \, i\in \Sigma$, satisfy the hypothesis of the previous theorem {since the marginalization of such points over all leaves except two has maximal rank.} Thus the hypothesis is still satisfied on a Zariski open subset containing $p$ and we are done.
\end{proof}

\begin{rmk}
The statement of the previous corollary also holds if we replace the thin flattening matrix $Tf_{A|B}(p)$ by a full-dimension block-diagonal flattening matrix, that is, the matrix obtained from $p\in  \LL^G = (W_A \otimes W_B)^G$ when rows and columns are indexed by basis of $W_A$ and $W_B$ consistent with the Maschke decomposition into isotypic components. {In this case, we have to replace $\mathbf{m}$ by $(d_1\,m_1,\dots,d_s\,m_s)$, see \eqref{ineq}.}
\end{rmk}

\begin{rmk}\rm As a consequence of Theorem \ref{thm_edges} we obtain that for  generic tensors  $p\in \LL^G$ in the image $\phi_T^G$ for a certain phylogenetic tree $T$ and group $G$, the rank conditions for the general Markov model (that is, $\rank flat_{A|B}(p)\leq \kappa$ for every edge split $A|B$ of $T$) are enough to {identify the variety $V_T$.}
Indeed, as proven in \cite{CF11}, for generic tensors $p$ in the union $X=\cup_T V_T$ over all trees with leaf set $L$, the rank conditions on the thin flattening are enough to {define} the variety $V_T$ to which $p$ belongs {inside $X$.} Now by Theorem \ref{thm_edges} we can translate these rank conditions into the easier condition of rank $\leq \kappa$, which can be directly tested in practice using the Eckart-Young Theorem \cite{Eckart1936} applied to the usual flattening matrix without dealing with the block structure or the irreducible representations of the group (so the rank conditions on the thin flattening do not add new information).
\end{rmk}

{On the other hand, Theorem \ref{thm_edges} also serves for obtaining explicit equations of $V_T^G$. This idea was hinted in section 3, but we formalize it here. Let $I_{A|B}$ be the ideal generated by all $(\kappa+1)\times (\kappa+1)$ minors of $flat_{A|B}(p)$, and let $J$ denote the intersection of the ideals of $m_k\times m_k$ minors of $flat_{i_a|j_b}(p')$, $k=1,\dots,s$ (using the notation of Theorem \ref{thm_edges}). Theorem \ref{thm_edges} implies the following result.}

{\begin{corollary}\label{cor:saturation}
    The saturation $(I_{A|B}+I(\LL^G)):J^{\infty}$ is contained in $I_T^G.$
\end{corollary} 
\begin{proof}
    Let $g$ be a polynomial in $\mathfrak{a}=(I_{A|B}+I(\LL^G)):J^{\infty}$. As $V_T^G$ is irreducible, it suffices to prove that $g(p)$ vanishes for any $p\in V_T^G\setminus V(J)$ to deduce that $g\in I_T^G$. Let $Z$ be the variety defined by $(m_{k}+1)\times (m_{k}+1)$ minors of $Tf_{A|B}(p)$, $k=1,\dots,s$. Note that $ V_T^G\setminus V(J)$ is contained in $Z\setminus V(J)$ and Theorem \ref{thm_edges} implies that this set coincides with $\left(V(I_{A|B})\cap \LL^G\right)\setminus V(J)$. Thus, if $p\in V_T^G\setminus V(J)$, $p$ also belongs to  $\left(V(I_{A|B})\cap \LL^G\right)\setminus V(J)$. Then $g(p)$ is zero because $g$ vanishes on the closure of this set (by properties of quotient ideals, the zero set of $\mathfrak{a}$ is the closure of $\left(V(I_{A|B})\cap \LL^G\right)\setminus V(J)$, see \cite{Cox1997}).
\end{proof} 
}

{For instance, in Example \ref{sec5_ExK81} we have $J=(\bp_{\texttt{AAAA}}\bp_{\texttt{ACAC}}\bp_{\texttt{AGAG}}\bp_{\texttt{ATAT}})$ and $\LL \setminus V(J)$ is the open set $\mathcal{O}$ of Section \ref{sec:k81}. Working with the quotient ideal of Corollary \ref{cor:saturation} is equivalent to restricting to $\mathcal{O}$ in terms of algebraic varieties.}

\subsection{Phylogenetic networks}

In this subsection we apply the previous results in the more general setting of \emph{binary phylogenetic networks} \cite[\S 10]{SteelPhylogeny},
that is, rooted acyclic directed graphs $\N$ (with no edges in parallel) satisfying: 1) the root $r$ has out-degree two, 2) every leaf has in-degree one, and 3) all other nodes have either in-degree one and out-degree two (these are called \emph{tree nodes}) or in-degree two and out-degree one (called \emph{reticulation nodes}). 

Following \cite{grosslong} and \cite{nakhleh2011}, we briefly recall the description of  Markov processes on phylogenetic networks and the corresponding notation.
A phylogenetic network is a tree-child network $\N$ whose set of leaves is in bijection with a finite set $L$. To model substitution of molecular units along a phylogenetic network one assigns a discrete random variable taking values in $\Sigma$ to each node on $\mathcal{N}$, then distribution $\pi $ is assigned to the root $r$, and each edge $e$ is assigned $\kappa\times \kappa$-transition matrix ${M^e}$ (both taken from the evolutionary model). Write $R = \{w_1, \dots , w_m\}$ for the set of reticulation nodes of $\N$, and denote by $e^0_i$ and $e^1_i$ the two edges directed into $w_i$. For $1\leq i\leq m$ assign a parameter $\delta_i\in (0,1)$ to $e^0_i$ and $1-\delta_i$ to $e^1_i$ so that with probability $\delta_i$ edge $e^0_i$ is removed and $e^1_i$ is kept (and with probability $1-\delta_i$ $e^0_i$ is kept and $e^1_i$ removed).
We write $\theta$ for the whole set of these substitution parameters.
%
%
Each binary vector $\sigma \in  \{0, 1\}^m$ encodes the possible choices for the reticulation edges, where $\sigma_i=0$ or $1$ means that the edge $e_i^0$ or $e_i^1$ is removed, respectively. Thus, each $\sigma \in  \{0, 1\}^m$ results in an $n$-leaf tree $T_{\sigma}$ (displayed by $\mathcal{N}$) rooted at $r$ with a collection of transition matrices corresponding to the particular edges that remain according to $\sigma$.
We call $\theta_{\sigma}$ the restriction of the substitution parameters $\theta$ of the network to $T_{\sigma}$.

{The \emph{displayed tree model} presents the distribution on the set of site-patterns  $\Sigma^n$  (or assignment of states at the leaves of $\mathcal{N}$)
as a mixture of distributions on the displayed trees of $\mathcal{N}$ as follows:}
\begin{eqnarray*}
P_{\N,\theta} = \sum_{\sigma\in \{0,1\}^m} \left (
\prod_{i=1}^m\delta_i^{1-\sigma_i} (1-\delta_i)^{\sigma_i}
\right ) \phi_{T_\sigma}(\theta_{\sigma})
\end{eqnarray*}

One can define it analogously if all parameters are taken from a $G$-equivariant model.

Assume that $N$ has a clade $T_A$, $A\subset L$ 
that does not contain any reticulation node (this is illustrated in the network of Figure \ref{figure:network}, where the clade $T_A$ corresponds to leaves
1 and 2). Then $T_A$ is a subtree of $\N$ shared by all $T_{\sigma}$ and the transition matrices at the edges of $T_A$ are also shared by all $\theta_{\sigma}$. Write $B$ for the leaves in $\N$ not in $A$.

Theorem 2 of \cite{CasFerBirkhauser} together with the results of section 4 give:


\begin{figure}
\begin{center}
 \includegraphics[scale=0.4]{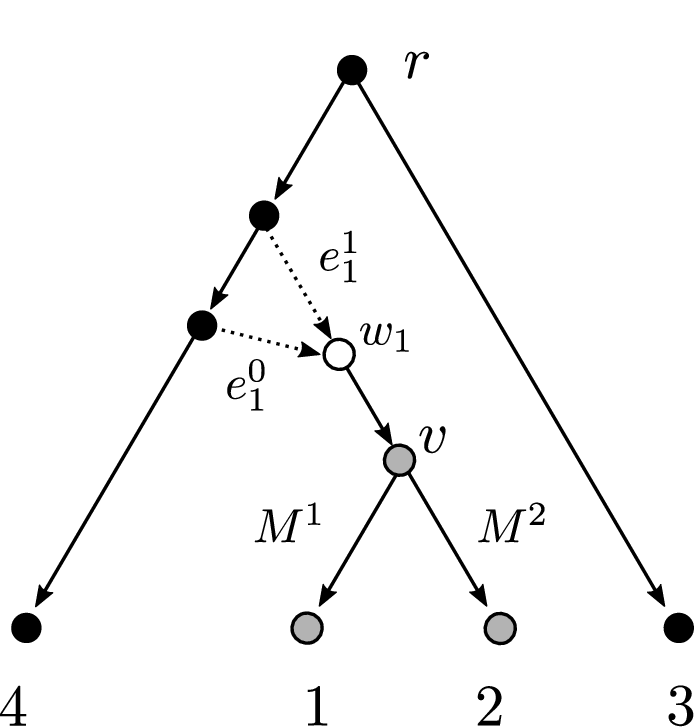}
 \caption{\label{figure:network} A 4-leaf phylogenetic network $\N$ with one reticulation node $w_1$ painted white. The clade corresponding to leaves $A = \{1, 2\}$ has been coloured with gray.}
 \end{center}
\end{figure}

\begin{thm}(\cite{CasFerBirkhauser})\label{thm_clade}
Consider a $G$-equivariant model on a phylogenetic network $\N$.
Assume that there is a clade $T_A$ in $\N$ that does not contain any reticulation node
and write $B=L\setminus A$.
If $p=P_{\N,\theta}$ is a distribution on $\N$, then
the block-rank of $Tf_{A|B}(p)$
is smaller than or equal to $\mathbf{m}$.
\end{thm}

{An analogous result can be obtained} if we consider the full-dimension block-diagonal flattening matrix instead of the thin flattening.

\begin{proof}
For the general Markov model this was proved in \cite{CasFerBirkhauser}. The proof for equivariant models follows from the general Markov model and Theorem \ref{thm_edges}.
\end{proof}

 {It seems plausible to generalize this theorem to other type of splits in phylogenetic networks (avoiding tree clades), but doing this in this paper would require introducing many notation for networks and we plan to address this in the forthcoming work \cite{CFGHS}.}
 
\section{Discussion and open questions}
Given a phylogenetic tree $T$ and a permutation group $G\leq \mathfrak{S}_{\kappa}$, we have investigated the connection between the algebraic variety associated to a  $G$-equivariant model on $T$ and the variety associated to the the general Markov model on the same tree. We have given a negative answer to Question 1 but
we have proved in Theorem \ref{thm_main} that $V_T^G$ is an irreducible component of $V_T\cap \LL^G$ for any trivalent tree $T$ and any group $G$.
As a consequence of the results, we have also seen that systems of phylogenetic invariants specific for $G$-equivariant models arise from rank $\kappa$ constraints applied to flattening matrices (if we take into account the $G$-invariance of the corresponding distributions).
This is true not only for trees but also for certain networks as shown in Section 5 (Theorem \ref{thm_clade}).

These theoretical results have a practical consequence: they imply that one can implement phylogenetic reconstruction methods based on phylogenetic invariants without the need of performing isotypical decompositions and based solely on the phylogenetic invariants of the general Markov model. For example, this implies that for the models K81, K80 and JC69 there is no need to apply a discrete Fourier transform on the data prior to applying algebraic methods.

In relation to Question 1 and motivated by the examples and results of section 3, we pose the following questions:
\begin{enumerate}
    \item For which trees and models does $V_T\cap \LL^G$ coincide with $V_T^G$? In other words, in which cases is this intersection an irreducible variety?
    \item In which cases is $I(V_T)+I(\LL^G)=I(V_T^G)$?
%
\end{enumerate}
In view of the examples of section 3, we conjecture that $V_T\cap \LL^G$ coincides with $V_T$ only when $T$ is a star tree; similarly we believe that for star trees it is natural to expect $I(V_T)+I(\LL^G)=I(V_T^G)$.

It is also natural to ask whether  Theorem \ref{thm_main} can be generalized to non-trivalent trees. If one wants to prove it with the same kind of arguments we used, then it would be enough to prove that generic points of no evolution are smooth points of $V_T$.

From a more practical point of view, in \cite{CFM} we provided equations for complete intersections that defined the varieties $V_T^G\subset \LL^G$ on certain open subsets containing the biologically relevant points. These equations were obtained by extending some equations from tripods and considering certain minors of the flattening matrices. From the work done here it is natural to expect that this procedure can be done by intersecting the complete intersection given for the general Markov model, with the corresponding space $\LL^G$. As observed in the examples of section 3, one has to take into account the decrease in the degree of equations obtained when imposing $G$-invariance to equations from the general Markov model.

Regarding Theorem \ref{thm_edges} and its implications for phylogenetic networks (Theorem \ref{thm_clade}), in a forthcoming work  \cite{CFGHS} we will explore the consequences on the identifiability of phylogenetic networks evolving under equivariant models.

\bibliographystyle{plain}

\newpage
\appendix
\section{Macaulay2 computations for section 3}
Here the notation of Small Phylogenetic Trees webpage \cite{Smalltrees} is adopted and ideals of tripod trees are obtained from this webpage.
\subsection{From K81 to K80}\label{sec:app1}
\subsubsection{Tripods}

The follwing M2 code is also available at \url{https://github.com/mcasanellas/Phyloinvariants}

\begin{verbatim}
R = QQ[q1,q2,q3,q4,q5,q6,q7,q8,q9,q10,q11,q12,q13,q14,q15,q16];

IK3 =ideal(q1*q8*q15-q3*q5*q16,q4*q6*q15-q2*q7*q16,
q7*q12*q14-q8*q10*q15,q1*q12*q14-q2*q9*q16,q5*q11*q14-q6*q9*q15,
q4*q11*q14-q3*q10*q16,q6*q12*q13-q5*q10*q16,q3*q12*q13-q4*q9*q15,
q8*q11*q13-q7*q9*q16,q2*q11*q13-q1*q10*q15,q2*q8*q13-q4*q5*q14,
q3*q6*q13-q1*q7*q14,q2*q8*q11-q3*q6*q12,q4*q5*q11-q1*q7*q12,
q2*q7*q9-q3*q5*q10,q4*q6*q9-q1*q8*q10,q1*q4*q14*q15-q2*q3*q13*q16,
q6*q8*q13*q15-q5*q7*q14*q16,q1*q6*q12*q15-q2*q5*q11*q16,
q4*q8*q11*q15-q3*q7*q12*q16,q11*q12*q13*q14-q9*q10*q15*q16,
q4*q6*q12*q14-q2*q8*q10*q16,q3*q5*q12*q14-q2*q8*q9*q15,
q1*q8*q11*q14-q3*q6*q9*q16,q2*q7*q11*q14-q3*q6*q10*q15,
q1*q8*q12*q13-q4*q5*q9*q16,q2*q7*q12*q13-q4*q5*q10*q15,
q4*q6*q11*q13-q1*q7*q10*q16,q3*q5*q11*q13-q1*q7*q9*q15,
q3*q8*q10*q13-q4*q7*q9*q14,q2*q6*q9*q13-q1*q5*q10*q14,
q5*q8*q10*q11-q6*q7*q9*q12,q2*q4*q9*q11-q1*q3*q10*q12,
q3*q4*q5*q6-q1*q2*q7*q8);

--LG: permute C and T
LG=ideal(q4-q2,q12-q10,q13-q5,q14-q8,q15-q7,q16-q6);
Inter2=IK3+LG;
mp=minimalPrimes Inter2
--output: a single ideal which coincides with the one for K80
\end{verbatim}

\subsubsection{Quartets}
The M2 code that computes the minimal primes of $I_T^{K81}+I(\LL^{K80})$ for subsection \ref{sec_quartetsK80} is available at:

\url{https://github.com/mcasanellas/Phyloinvariants}

\subsection{From K80 to JC69}\label{sec:app2}
\subsubsection{Tripods}
The follwing M2 code is also available at \url{https://github.com/mcasanellas/Phyloinvariants}
\begin{verbatim}
R = QQ[q1,q2,q3,q4,q5,q6,q7,q8,q9,q10];
IK2 =ideal(q1*q6*q9-q2*q4*q10,q3*q5*q9-q2*q7*q10,q5*q6*q8-q4*q7*q10,
q2*q6*q8-q3*q4*q9,q2*q5*q8-q1*q7*q9,q3*q4*q5-q1*q6*q7,
q1*q3*q9^2-q2^2*q8*q10,q1*q6^2*q8-q3*q4^2*q10,q3*q5^2*q8-q1*q7^2*q10);

-- LG: permute G with C and G with T
LG=ideal(q2-q3,q6-q7,q7-q9,q4-q8,q5-q10);
Inter=IK2+LG;
minimalPrimes Inter
--output. a unique minimal prime which coincides with JC69:
--{ideal(q7-q9,q6-q9,q5-q10,q4-q8,q2-q3,q1*q9^2-q3*q8*q10)}
\end{verbatim}

\end{document}